\documentclass[11pt,a4paper]{article}
\usepackage[T1]{fontenc}
\usepackage[utf8]{inputenc}
\usepackage{lmodern}
\usepackage[margin=28mm]{geometry}
\usepackage{amsmath,amssymb,amsthm,mathtools}
\usepackage{microtype}
\usepackage{enumitem}
\usepackage{needspace}
\usepackage{fancyhdr}
\usepackage[hyphens]{url}
\usepackage[hidelinks,pdfencoding=auto]{hyperref}
\usepackage[capitalise,noabbrev]{cleveref}

\newcommand{\class}[1]{\mathsf{#1}}
\newcommand{\Pclass}{\class{P}}
\newcommand{\NP}{\class{NP}}
\newcommand{\coNP}{\class{coNP}}
\newcommand{\PH}{\class{PH}}
\newcommand{\SAT}{\class{SAT}}
\newcommand{\CSAT}{\class{CircuitSAT}}
\newcommand{\UNSAT}{\class{UNSAT}}
\newcommand{\UCSAT}{\class{UniqueCircuitSAT}}
\newcommand{\DTIME}{\class{DTIME}}
\newcommand{\RTIME}{\class{RTIME}}
\newcommand{\BPTIME}{\class{BPTIME}}
\newcommand{\ZPP}{\class{ZPP}}
\newcommand{\EXP}{\class{EXP}}
\newcommand{\NEXP}{\class{NEXP}}
\newcommand{\REXP}{\class{REXP}}
\newcommand{\UEXP}{\class{UEXP}}
\newcommand{\EXPH}{\class{EXPH}}
\newcommand{\BPEXP}{\class{BPEXP}}
\newcommand{\QH}{\class{QH}}
\newcommand{\BPQP}{\class{BPQP}}

\newcommand{\poly}{\mathrm{poly}}
\newcommand{\mc}{\text{-}\class{mc}}
\newcommand{\bits}{\{0,1\}}
\newcommand{\F}{\mathbb F}
\newcommand{\E}{\mathbb E}
\newcommand{\ind}{\mathbf 1}
\newcommand{\Sclass}{\class{S}_2}
\newcommand{\Llog}{\lambda}
\newcommand{\hyp}{\mathcal H_d}
\newcommand{\Hrel}{\hyp^{\mathrm{rel}}}
\newcommand{\Hsr}{\hyp^{\mathrm{sr}}}
\newcommand{\Hwr}{\hyp^{\mathrm{wr}}}
\newcommand{\Td}{T_d}
\DeclareMathOperator{\Sound}{Sound}
\DeclareMathOperator{\Complete}{Complete}

\theoremstyle{plain}
\newtheorem{theorem}{Theorem}[section]
\newtheorem{lemma}[theorem]{Lemma}
\newtheorem{proposition}[theorem]{Proposition}
\newtheorem{corollary}[theorem]{Corollary}
\theoremstyle{definition}
\newtheorem{definition}[theorem]{Definition}
\newtheorem{hypothesis}[theorem]{Hypothesis}
\theoremstyle{remark}
\newtheorem{remark}[theorem]{Remark}
\crefname{hypothesis}{Hypothesis}{Hypotheses}
\crefname{lemma}{Lemma}{Lemmas}
\crefname{proposition}{Proposition}{Propositions}
\crefname{corollary}{Corollary}{Corollaries}

\setlist[enumerate]{topsep=5pt,itemsep=3pt}
\numberwithin{equation}{section}

\hypersetup{
  pdftitle={Consequences of Polylogarithmic Membership Comparability for SAT},
  pdfauthor={Sebastian Ben Daniel},
  pdfsubject={Expanded version: relative-gap advice, SAT simulations, and oracle limits},
  pdfkeywords={membership comparability, SAT, polynomial hierarchy, advice, isolation}
}

\title{\LARGE\bfseries Consequences of Polylogarithmic\\
Membership Comparability for SAT}
\author{Sebastian Ben Daniel}
\date{September 26, 2026\\[3pt]\small Expanded version}

\begin{document}
\maketitle
\thispagestyle{plain}

\begin{abstract}
We study the consequences of membership comparators that exclude one
possible membership vector, deterministically or with a relative
advantage over uniform guessing. For every polynomially bounded arity,
a randomized polynomial-time comparator of error at most
$(1-1/\poly(n))2^{-t}$ gives
$\NP/\poly\cap\coNP/\poly$ recognition with common advice. The proof
uses limited independence, polynomial occurrence certificates, and a
self-contained positive-relation advice transfer.
For SAT at arity $O((\log n)^d)$, both this relative-gap hypothesis and
deterministic comparability imply $\PH=\Sclass^{\NP}$, the uniform
bound $\PH\subseteq\BPTIME(2^{O((\log n)^{d^2})})$, and symmetric
verification with polynomial-length certificates and an oracle-free
deterministic $2^{O((\log n)^d)}$ predicate. Polynomial-advice
deterministic decoding has the same exponent $d$.
Applying the randomized simulation to an unconditional diagonal language
yields, for every fixed $\varepsilon>0$,
$\mathrm{BPP}\subsetneq\BPTIME(2^{O((\log n)^{d^2+\varepsilon})})$,
without advice. The larger clock remains subexponential under every
fixed number of self-compositions. A layered oracle satisfies
deterministic comparability and $\NP^O=\coNP^O$ but excludes randomized
NP algorithms with smaller logarithmic power, establishing a relativized
limit on the SAT exponent $d$.
This expanded version also develops the full weak-advantage regime,
where saving $2^{-O((\log n)^d)}$ gives randomized SAT exponent $d$
and PH exponent $d^k$ at fixed level $k$; the quasipolynomial and
exponential hierarchy consequences; binary-comparator advice bounds;
and the certificate-length boundary between the randomized regimes.
Under deterministic comparability, uniform deterministic promise-unique
search additionally gives $\UEXP=\EXP$. The ordinary second-level
collapse $\PH=\Sigma_2^p$ for $d>1$ remains unproved.
\end{abstract}

\clearpage
\tableofcontents
\clearpage
\section{Scope, conventions, and main conclusions}
\label{sec:intro}

This paper studies randomized membership comparability beyond the
absolute-advantage regime. Its general relative-gap theorem applies to
arbitrary languages at polynomial arity; its SAT consequences track
uniform time, predicate time, witness length, and ordinary advice
separately. The reconstruction builds on Sivakumar~\cite{Siv99}; the
advice and isolation questions also draw on Ben Daniel's thesis
\cite{BD13} and binary-certificate manuscript~\cite{BD26}. The imported
structural results are the advice theorem of Amir, Beigel, and Gasarch
\cite{ABG03} and the strengthened Yap collapse of Cai, Chakaravarthy,
Hemaspaandra, and Ogihara~\cite{CCHO05}.

\paragraph{Relationship to the submission manuscript.}
This is the expanded version of
\emph{Relative-Gap Membership Comparability: Advice, SAT Simulations,
and Oracle Limits}. That focused manuscript contains the general
relative-gap theorem, its self-contained positive-relation proof, the
polylogarithmic SAT consequences, the bounded-error time separation,
and the oracle lower bound. The present version retains the full
weak-advantage analysis, the additional hierarchy and advice
consequences, and the examination of the remaining second-level
question. The two versions contain overlapping results from the same
work, not disjoint claims intended as independent submissions.

\Cref{sec:nd-hard-extension} uses an unconditional Kannan-style
truth-table diagonalization~\cite{kannan82}, adapted to signed
existential circuits. The languages $H$ and $G_h$, their fifth-level
alternating-time characterizations, and their circuit hardness are
independent of $\hyp$. Their randomized clocks, and the separate advice
containment $\PH\subseteq\NP/\poly$, are the two conditional inputs.
The diagonalization and the qualitative implication
$\QH\subseteq\BPQP\Rightarrow\mathrm{BPP}\subsetneq\BPQP$ are standard
background; the contributions recorded here are the quantitative
simulations, their exponent accounting, and the layered oracle bound.

The time separation has a specific position among hierarchy results.
Karpinski--Verbeek's padding argument separates BPP from bounded-error
subexponential time~\cite{karpinski-verbeek87}. Lu, Oliveira, and
Santhanam identify the smaller-gap question of separating
$\BPTIME(n)$ from $\BPTIME(T(n))$ when every fixed self-composition of
$T$ is still subexponential~\cite{lu-oliveira-santhanam21}.
Our conditional quasipolynomial clock falls in that regime.
\Cref{sec:oracle} supplies a complementary limit on relativizing
improvements to the randomized SAT clock.

All logarithms are base two. Write
\[
  \Llog(n)=\lceil\log_2(n+2)\rceil.
\]
The notation $\DTIME(2^{O(\Llog(n)^a)})$, and its randomized analogues,
denotes the union over constant multiplicative factors in the exponent.
Constants may depend on the language and the fixed algorithms, but not on
the input length. Replacing $\Llog(n)$ by $\log n$ gives the same asymptotic
classes. Algorithms for deterministic and bounded-error time containments are
uniformly clocked on every input and every random tape. Randomized one-sided error means that no-instance acceptance has
probability zero and yes-instance acceptance has probability at least $2/3$.
Bounded two-sided error means error at most $1/3$ on every input.

\begin{definition}[Membership comparability]
\label[definition]{def:membership-comparability}
A language $A$ is $k(n)$-membership comparable if a deterministic
polynomial-time algorithm, on an ordered tuple $(x_1,\ldots,x_t)$ with
$t\geq k(n)$ and $n=\max_i|x_i|$, outputs $b\in\bits^t$ such that
\[
 b\neq \bigl(\chi_A(x_1),\ldots,\chi_A(x_t)\bigr).
\]
The running time is polynomial in the total tuple encoding length.
\end{definition}

This is Sivakumar's maximum-individual-length convention~\cite{Siv99}.
Larger arity is a weaker hypothesis. Throughout, comparisons use the
threshold formulation just defined.

\begin{hypothesis}[Polylogarithmic comparability, $\hyp$]
\label[hypothesis]{hyp:main}
Fix an integer $d\geq 1$. There are a constant integer $c_0\geq 1$ and one
deterministic polynomial-time algorithm $g$ that compares every tuple of
SAT instances whose arity is at least $c_0\Llog(n)^d$, where $n$ is the
maximum encoding length of its formulas.
\end{hypothesis}

We assume standard explicit encodings of formulas and circuits. Every
named circuit input contributes to its encoding, and circuits have at
most their encoding length many input variables. Malformed encodings are
rejected. Reductions use uniquely determined gate variables when
preservation of the number of satisfying assignments is required.

\begin{definition}[The promise-unique problem]
An algorithm solves $\UCSAT$ if it rejects every circuit with zero
satisfying assignments and accepts every circuit with exactly one
satisfying assignment. Either answer is allowed on a circuit with more
than one satisfying assignment. Its time bound must hold on all circuits,
including the off-promise ones.
\end{definition}

\paragraph{Polynomial-length alternation with a longer predicate clock.}
For a time bound $T$, write $\Sigma_2^p[T]$ for the languages having a
characterization
\[
 x\in L\iff
 \exists y\in\bits^{p(n)}\ \forall z\in\bits^{q(n)}\ P(x,y,z)=1,
 \qquad n=|x|,
\]
where $p,q$ are fixed polynomials and the deterministic predicate is
clocked in time $O(T(n))$ on all permitted triples. Time is measured in
the original input length $n$, not in an expanded circuit description.
Define $\Pi_2^p[T]$ by reversing the two quantifiers. Bounds written as
$2^{O(\Llog(n)^a)}$ take a union over the constant in the exponent.
Thus $\Sigma_2^p[\poly]=\Sigma_2^p$, whereas
$\Sigma_2\mathrm{TIME}(T)$ may use quantified strings as long as $T$.
These conventions distinguish the resource bounds in
\cref{thm:certified-collapse,thm:s2-simulation} from those in \cref{thm:qh}.
The symmetric class $\Sclass^p[T]$, defined in
\cref{def:s2-clock}, additionally requires a winning certificate for the
appropriate prover against all opposing certificates.

For the larger time scales, define
\begin{align*}
 \QH&=\bigcup_{j,c\geq 1}
       \Sigma_j\mathrm{TIME}\bigl(2^{O(\Llog(n)^c)}\bigr),\\
 \BPQP&=\bigcup_{c\geq1}
       \BPTIME\bigl(2^{O(\Llog(n)^c)}\bigr).
\end{align*}
Here $j$ and $c$ are fixed integers in each constituent class.
$\Sigma_j\mathrm{TIME}$ has at most $j$ alternating quantifier blocks,
beginning existentially, with the indicated time bound. In particular,
$\QH$ does not allow an unbounded number of alternations.
We use
\[
 \EXP=\bigcup_{c\geq1}\DTIME(2^{n^c}),\quad
 \NEXP=\bigcup_{c\geq1}\mathrm{NTIME}(2^{n^c}),\quad
 \REXP=\bigcup_{c\geq1}\RTIME(2^{n^c}).
\]
$\UEXP$ is the corresponding union for nondeterministic machines with at
most one accepting computation on every input.

\paragraph{Unconditional diagonalization.}
There are fixed languages
\[
 H\in\Sigma_5\mathrm{TIME}(2^{O(n)}),\qquad
 G\in\QH\setminus(\NP/\poly\cup\coNP/\poly),
\]
where every length-$n$ slice of $H$, for $n\ge6$, requires more than
$\lfloor2^n/(10n)\rfloor$ gates in both the nondeterministic and
conondeterministic models of \cref{sec:nd-hard-extension}.
The language $G$ is the prefix projection with
$\ell(n)=\Llog(n)\max\{1,\lceil\log_2\Llog(n)\rceil\}$.
\Cref{lem:unconditional-nd-hard,cor:unconditional-qh-hard} prove this
statement without a membership-comparability hypothesis. In the next
theorem, $H$ and $G$ always denote these fixed languages.

\begin{theorem}[Consequences of $\hyp$]
\label{thm:main}
Under \cref{hyp:main}, the following hold.
\begin{enumerate}[label=(\roman*)]
 \item $\UCSAT$ has a deterministic algorithm with running time
       $2^{O(\Llog(n)^d)}$, and
       $\SAT,\CSAT\in\RTIME(2^{O(\Llog(n)^d)})$.
 \item
 \begin{equation}
   \PH\subseteq
   \BPTIME\bigl(2^{O(\Llog(n)^{d^2})}\bigr).
   \label{eq:main-ph}
 \end{equation}
 The exponent $d^2$ does not depend on the level of $\PH$.
 \item $\PH\subsetneq\QH=\BPQP$.
 \item $\NEXP=\REXP$ and $\UEXP=\EXP$.
 \item Every language in $\PH$ has a symmetric-alternation
       characterization with polynomial-length certificates and a
       deterministic predicate of exponent $d$:
 \begin{equation}
 \begin{split}
  \PH&\subseteq\Sclass^p\!\left[2^{O(\Llog(n)^d)}\right]\\
     &\subseteq
       \Sigma_2^p\!\left[2^{O(\Llog(n)^d)}\right]\cap
       \Pi_2^p\!\left[2^{O(\Llog(n)^d)}\right].
 \end{split}
 \label{eq:main-two-block}
 \end{equation}
 \item $\PH\subseteq\DTIME(2^{O(\Llog(n)^d)})/\poly$.
 \item $\EXPH=\BPEXP$.
  \item The fixed language $H$ belongs to
        $\BPTIME(2^{O(n^{d^2})})$; hence
        $\BPEXP\not\subseteq\NP/\poly\cup\coNP/\poly$.
        For every fixed positive integer $k$,
        $\BPTIME(2^{O(\Llog(n)^{d^2})})\not\subseteq\mathrm{SIZE}(n^k)$.
  \item The fixed language $G$ belongs to
        \[
        \BPTIME\!\left(2^{O((\Llog(n)\log\Llog(n))^{d^2})}\right),
        \]
        and therefore to
        $\BPTIME(2^{O(\Llog(n)^{d^2+\varepsilon})})$
        for every fixed $\varepsilon>0$.
 \item For every fixed $\varepsilon>0$,
       \[
       \mathrm{BPP}\subsetneq
       \BPTIME\!\left(2^{O(\Llog(n)^{d^2+\varepsilon})}\right).
       \]
       The separation concerns uniformly decided languages with no advice.
 \item
       \[
       \PH/\poly=\NP/\poly=\coNP/\poly,
       \qquad \BPQP\not\subseteq\PH/\poly.
       \]
\end{enumerate}
\end{theorem}

In particular, an $O(\log^2 n)$-arity comparator gives
\[
  \PH\subseteq\BPTIME\bigl(2^{O((\log n)^4)}\bigr).
\]
For $d=1$, the unique-solution and PH simulation bounds are polynomial;
this is consistent with the known logarithmic-arity randomized
consequences~\cite{Siv99}. The hard-language and time-separation clocks
retain their separately displayed overheads.
For $d>1$, the predicate in \eqref{eq:main-two-block} is still allowed
quasipolynomial time. Its quantified strings remain polynomial, and its
exponent is $d$, not $d^2$. The proposed strengthening
\begin{equation}
 \SAT\in O(\log^2n)\mc
 \quad\Longrightarrow\quad \PH=\Sigma_2^p
 \label{eq:target}
\end{equation}
is \emph{not proved in this note}. Section~\ref{sec:gap} first diagnoses
the refutation-advice completeness obstruction and then avoids that test
by using certified satisfiability search. Section~\ref{sec:learning}
identifies the additional small-circuit condition equivalent to replacing
this search step by polynomial time; it does not claim that every possible
proof of \eqref{eq:target} must establish that condition.

\paragraph{Relativized lower bound.}
For each fixed integer $d\ge2$ and $c_0\ge1$, \cref{cor:oracle-hd}
constructs an oracle $A$ for which the threshold
$\lceil c_0\Llog(n)^d\rceil$ holds for the relativized complete set,
$\NP^A=\coNP^A$, and
\[
 \NP^A\not\subseteq\BPTIME^A\!\left(2^{\delta\Llog(n)^d}\right)
 \qquad(0<\delta<c_0).
\]
The comparison concerns the logarithmic power in the randomized SAT
upper bound. In this oracle world $\PH^A=\NP^A$, so the result is
compatible with the proposed second-level collapse.

\paragraph{Two regimes for randomized comparators.}
Write $\mathrm{err}_g(\tau)\le(1-\gamma(n))2^{-t}$, where $t$ is the
arity and $n$ the maximum individual-instance length. The baseline
$2^{-t}$ is the error of a uniform output. The two regimes treated in
\cref{sec:randomized,sec:relative-gap} are distinguished by the size of
the guaranteed relative saving $\gamma$.

If $\gamma(n)\ge1/\poly(n)$, the relative-gap hypothesis $\Hrel$ gives
all conclusions of \cref{thm:main} except the uniform deterministic
promise-unique algorithm and $\UEXP=\EXP$. In particular, it gives
$\PH=\Sclass^{\NP}$, the level-independent randomized PH exponent
$d^2$, and the polynomial-length symmetric simulation with predicate
exponent $d$. The deterministic binary-comparator bound
\eqref{eq:binary-advice} uses polynomial advice in this regime; its
one-sided randomized version retains linear advice.

If only $\gamma(n)\ge2^{-O(\Llog(n)^d)}$ is guaranteed, the weak
hypothesis $\Hwr$ still yields one-sided randomized SAT exponent $d$,
fixed-level PH exponent $d^k$, and the larger-time hierarchy equalities
and time separations of \cref{thm:wr}. For $d>1$, whether $\Hwr$
implies $\NP\subseteq\coNP/\poly$ remains unresolved here.

The distinction is the length of a heavy-output certificate. With
$b_n$ the tuple-code length, the construction uses
\[
 R=\Theta\bigl((b_n+1)/\gamma(n)^2\bigr),\qquad
 J=2^{k(n)}R,\qquad |\pi|=\Theta(R\log J).
\]
Limited independence gives polynomial-size sample-table descriptions in
both regimes. In the first, listing $R$ occurrence indices gives an
$\NP/\poly$ comparison relation, to which the full positive-relation
proof in Appendix~\ref{app:positive-abg} applies. In the second, the
occurrence list can be quasipolynomially long, so this structural
transfer is unavailable even though advised reconstruction still has
exponent $d$. At $d=1$, the weak saving is already inverse-polynomial
and the distinction disappears.

\section{A parameter-explicit promise-unique algorithm}
\label{sec:unique}

We record the two combinatorial and algebraic tools in the form required
for the rescaling. The reconstruction theorem is an imported algorithmic
result; the rest of the parameter calculation is given explicitly.

\begin{lemma}[Sauer--Shelah shortlisting; see~\cite{Sauer72,Siv99}]
\label[lemma]{lem:sauer}
Let $\mathcal S\subseteq\bits^m$. If no set of $K$ coordinates is
shattered by $\mathcal S$, where $1\leq K\leq m$, then
\[
 |\mathcal S|\leq\sum_{i=0}^{K-1}\binom mi.
\]
In particular, the bound applies when, for every $K$-subset of coordinates,
one prescribed pattern on that subset is forbidden.
\end{lemma}

\begin{lemma}[Sudan reconstruction; \cite{Sud96}, as used in~\cite{Siv99}]
\label[lemma]{lem:sudan}
Given a finite field $F$, an integer $D\geq1$, and $L$ distinct pairs in
$F\times F$, a deterministic algorithm running in time polynomial in
$|F|,D,L$ produces a polynomial-size list containing every polynomial
$p\in F[X]$ of degree at most $D$ that agrees with more than
$\sqrt{2LD}$ of the given pairs. Repeated first coordinates are permitted;
only the pairs must be distinct.
\end{lemma}

\begin{proposition}[Rescaled unique-solution reconstruction]
\label[proposition]{prop:unique}
Under $\hyp$, $\UCSAT$ has a deterministic algorithm with running time
$2^{O(\Llog(N)^d)}$ on circuits of encoding length $N$. Moreover, there is
a deterministic search version $U'$ with the same clock that returns
an assignment or $\bot$. On every circuit, any returned assignment
satisfies that circuit; on every uniquely satisfiable circuit, $U'$
returns its unique satisfying assignment. On a satisfiable off-promise
circuit it may return either a genuine witness or $\bot$.
\end{proposition}

\begin{proof}
Let $C$ be the input circuit, with $w\leq N$ input variables. A circuit
with $w=0$ is evaluated directly; the search version returns the empty
assignment if its value is one and $\bot$ otherwise. Malformed encodings
are rejected, and their search output is $\bot$. Otherwise, for
$a\in\bits^w$ define
\[
 P_a(X)=\sum_{i=0}^{w-1}a_iX^i.
\]
Choose an integer
\begin{equation}
 m=\lceil A\Llog(N)^d\rceil,\qquad q=2^m,
 \label{eq:field-size}
\end{equation}
where the constant $A$ will be fixed sufficiently large.

\paragraph{Field construction and the auxiliary NP language.}
Construct a monic irreducible polynomial $h\in\F_2[X]$ of degree $m$,
and represent $F=\F_{2^m}$ as $\F_2[X]/(h)$. Even brute-force
construction costs only $2^{O(m)}$ time: enumerate monic degree-$m$
polynomials and test divisibility by monic polynomials of degrees at most
$\lfloor m/2\rfloor$. An irreducible polynomial exists in every degree.

To ensure that the auxiliary language is in $\NP$, we pass $h$ as part
of each auxiliary input; that language does \emph{not} perform the
exponential-time field construction. On a well-formed tuple $(C,h,u,j)$,
with $h$ monic, $\deg h=m$, $u\in\bits^m$, and $0\leq j<m$, let
\begin{equation}
 (C,h,u,j)\in B
 \quad\Longleftrightarrow\quad
 \exists a\in\bits^w\,
 \bigl[C(a)=1\ \land\ \operatorname{bit}_j(P_a(u)\bmod h)=1\bigr].
 \label{eq:bit-language}
\end{equation}
For the definition of $B$, arithmetic modulo any monic $h$ is enough;
irreducibility is only needed for the actual field chosen by the outer
algorithm. Evaluating the quotient-ring expression and $C(a)$ takes
polynomial time in the tuple length, so $B\in\NP$.

Reduce \eqref{eq:bit-language} to SAT by a fixed polynomial-time
many-one reduction, producing formulas $\phi_{u,j}$. Since
$m=O(\Llog(N)^d)\leq N$ for sufficiently large $N$, there is a fixed
integer $b$ such that
\begin{equation}
  |\phi_{u,j}|\leq N^b
  \qquad(u\in F,
  \ 0\leq j<m).
 \label{eq:constructed-size}
\end{equation}
The exponent $b$ is determined by the fixed auxiliary verifier and
reduction. Increasing a finite lower-size cutoff absorbs constants.
This is the point at which the size of the \emph{constructed} SAT
instances is taken into account.

\paragraph{Comparing coordinate subsets.}
Choose
\[
 K=\left\lceil c_0\Llog(N^b)^d\right\rceil.
\]
This arity is covered by $\hyp$ on every tuple of the formulas in
\eqref{eq:constructed-size}. Because $K=O(\Llog(N)^d)$, choose $A$
in \eqref{eq:field-size} large enough that, for all sufficiently large
$N$,
\begin{equation}
 1\leq K\leq m/32.
 \label{eq:arity-field-ratio}
\end{equation}
For a fixed $u\in F$, apply the comparator to the formulas indexed by
every $K$-subset of $\{0,\ldots,m-1\}$, using increasing coordinate
order. Each answer forbids one pattern on those coordinates. Let
$S_u\subseteq F$ be all $m$-bit field values satisfying every such
exclusion.

If $C$ has the unique satisfying assignment $a$, the actual membership
bits of $\phi_{u,0},\ldots,\phi_{u,m-1}$ are exactly the coordinates of
$P_a(u)$. Consequently $P_a(u)\in S_u$. Regardless of whether $C$
is promised, \cref{lem:sauer} bounds the list size. With $H_2$ denoting
binary entropy,
\begin{equation}
 |S_u|\leq\sum_{i=0}^{K-1}\binom mi
 \leq 2^{H_2(1/32)m}<2^{m/3}=q^{1/3}.
 \label{eq:shortlist}
\end{equation}
Here $H_2(1/32)<1/3$, and the usual entropy estimate bounds the lower
binomial tail. The strict numerical constant is inessential; a fixed
exponent below one suffices after adjusting the other parameters.

\paragraph{Reconstructing an assignment.}
Form the distinct pairs
\[
 \mathcal D=\{(u,v):u\in F,
                  \ v\in S_u\}.
\]
Their number $L$ is at most $q^{4/3}$. If $L=0$, reject (return $\bot$
in the search version). Otherwise apply
\cref{lem:sudan} with degree bound $D=N$. Increase $A$, if necessary,
so that
\begin{equation}
 q^{2/3}>2N,
 \qquad\text{and hence}\qquad
 q>\sqrt{2Nq^{4/3}}\geq\sqrt{2NL}.
 \label{eq:decoding-condition}
\end{equation}
For a uniquely satisfiable $C$, the polynomial $P_a$ has degree at most
$w-1\leq N$ and agrees with all $q$ pairs $(u,P_a(u))$. It must therefore
appear in the output list.

For each returned polynomial, test whether its coefficients beyond
position $w-1$ vanish and its first $w$ coefficients belong to
$\{0,1\}\subseteq F$. When they do, interpret them as an assignment and
check it directly in $C$. The search algorithm $U'$ returns the first
assignment that passes this check, in a fixed deterministic order, and
returns $\bot$ if none passes. The decision version accepts exactly when
$U'$ returns an assignment. Thus every unsatisfiable circuit is rejected,
every uniquely satisfiable circuit yields its unique witness, and even
off-promise outputs are certified by direct evaluation of $C$; on a
satisfiable nonunique input, $U'$ may also return $\bot$.

\paragraph{Uniform running time.}
Field construction costs $2^{O(m)}$. There are $q$ field points, $m$
auxiliary formulas per point, and at most $2^m$ coordinate subsets per
point. Explicitly checking every candidate field value against every
subset costs at most $2^{O(m)}\poly(N)$, even without optimizing this
enumeration. Reconstruction and candidate verification are polynomial
in $q,N,L$, hence also cost $2^{O(m)}\poly(N)$. Since $d\geq1$, this is
$2^{O(\Llog(N)^d)}$. All loops have these bounds on every input,
including off-promise inputs. Finitely many shorter inputs are handled
by exhaustive search, returning the first verified witness or $\bot$
in the search version.
\end{proof}

\begin{remark}[What is being rescaled]
The polynomial-evaluation encoding, coordinate shortlisting, and list
reconstruction are Sivakumar's method~\cite{Siv99}. The preceding proof
spells out its use at field degree $\Theta(\Llog(N)^d)$ and preserves the
exponent $d$ despite the SAT reductions. It does not assume that a
superpolynomial-sized circuit has a polynomial-sized explicit encoding.
\end{remark}

\section{Uniform isolation without advice}
\label{sec:isolation}

We use the constant-success, all-prefix isolation formulation from
\cite[Theorem~7.1 and Lemma~F.4]{BD26}. It is a parameter-preserving
application of the isolation method of Valiant and Vazirani~\cite{VV86}.

\begin{lemma}[Linear-seed isolation]
\label[lemma]{lem:isolation}
Let $U$ solve the promise problem $\UCSAT$ deterministically in time
$2^{O(\Llog(N)^d)}$ on circuits of length $N$. There is a randomized
algorithm for $\CSAT$ with one-sided error and the same asymptotic time
bound. Before constant amplification, its acceptance probability on a
satisfiable circuit is at least $3/16$.
\end{lemma}

\begin{proof}
Let $C$ have $w\geq1$ inputs; the case $w=0$ is direct. Choose a uniform
binary Toeplitz matrix $T\in\F_2^{(w+1)\times w}$ and an independent
uniform $b\in\F_2^{w+1}$. There are $2w$ Toeplitz diagonals, so the seed
$s=(T,b)$ has length $3w+1$. Put $h_s(z)=Tz+b$.

For nonzero $v\in\F_2^w$, the $w+1$ coordinate linear forms of $Tv$
are independent as forms in the Toeplitz diagonals. To see this, index
the diagonal variables by row minus column and let $j_0$ be the least
index with $v_{j_0}=1$. In any nonempty linear combination of rows, the
largest selected row index has a highest-index diagonal variable that
cannot occur in any earlier selected row. The combination is therefore
nonzero. Thus $Tv$ is uniform. The independent shift implies that, for
$z\neq z'$, $h_s(z)$ and $h_s(z')$ are independent uniform vectors.

For every $r=1,\ldots,w+1$, construct
\[
 C_{s,r}(z)=C(z)\land[h_s(z)_{1\ldots r}=0^r]
\]
and return the OR of $U(C_{s,r})$ over all $r$. When $C$ is
unsatisfiable, every restricted circuit is unsatisfiable, and every
answer is zero.

Suppose the satisfying set has size $M>0$. Choose, for the analysis,
$r=\lceil\log M\rceil+1$, which lies in $\{1,\ldots,w+1\}$. If $Z$
is the number of assignments surviving this restriction, pairwise
independence gives
\[
 \mu=\E Z=M2^{-r}\in(1/4,1/2],
 \qquad \E[Z(Z-1)]\leq\mu^2.
\]
Since $\ind[Z=1]\geq Z-Z(Z-1)$ for each nonnegative integer $Z$,
\[
 \Pr[Z=1]\geq\mu-\mu^2\geq3/16.
\]
On this event the promise solver accepts this restricted circuit. An
acceptance on a nonunique restricted circuit is harmless, because the
original circuit is then satisfiable. Testing every prefix avoids the
factor-$w$ success loss incurred by selecting a single prefix length.

Every restricted circuit has size polynomial in $|C|$, and there are
polynomially many calls. The time bound is therefore preserved. A fixed
constant number of independent seeds, accepting if any trial accepts,
raises the success probability to at least $2/3$.
\end{proof}

\begin{corollary}
\label[corollary]{cor:random-sat}
Under $\hyp$,
\[
 \SAT,\CSAT\in\RTIME\bigl(2^{O(\Llog(n)^d)}\bigr).
\]
No advice is used. Every random seed in the algorithm is actually sampled.
\end{corollary}

\section{From a SAT algorithm to a second-level simulation}
\label{sec:second}

\begin{lemma}[Common-random-tape simulation]
\label[lemma]{lem:second}
If $\CSAT\in\RTIME(2^{O(\Llog(n)^d)})$ for a fixed integer $d\geq1$,
then
\begin{equation}
 \Sigma_2^p\cup\Pi_2^p
 \subseteq\BPTIME\bigl(2^{O(\Llog(n)^{d^2})}\bigr).
 \label{eq:second-level}
\end{equation}
\end{lemma}

\begin{proof}
Let $L\in\Sigma_2^p$. There are a polynomial $p$ and a polynomial-time
predicate $R$ such that, for $n=|x|$,
\begin{equation}
 x\in L\iff
 \exists u\in\bits^{p(n)}\ \forall v\in\bits^{p(n)}\ R(x,u,v)=1.
 \label{eq:sigma-two-form}
\end{equation}
Set $B(x,u)=\ind[\forall v\ R(x,u,v)=1]$. A circuit computing
$\neg R(x,u,\cdot)$ has polynomial size in $n$. Complementing the
assumed randomized circuit-satisfiability algorithm gives a bounded-error
algorithm for $B$ in time $2^{O(\Llog(n)^d)}$.

Amplify this algorithm so that, on every pair $(x,u)$, its error is at
most $2^{-p(n)-6}$. This requires only $O(p(n)+1)$ repetitions and
preserves the time bound. Fix a uniform clock and pad the coin usage to
one common length. Write the resulting computation as
$\widehat B(x,u;r)$, where $r$ is its random tape.

For a fixed $x$, take one common random tape for all possible $u$.
Independence between the answers for different $u$ is unnecessary:
a union bound gives
\begin{equation}
 \Pr_r\bigl[\forall u\in\bits^{p(n)},
                 \ \widehat B(x,u;r)=B(x,u)\bigr]
 \geq1-2^{p(n)}2^{-p(n)-6}=1-1/64.
 \label{eq:common-tape}
\end{equation}

For the sampled $r$, uniformly compile the deterministic computation
$u\mapsto\widehat B(x,u;r)$ into an explicit circuit $D_{x,r}$.
Hardwire $x$ and $r$; only $u$ is free. A standard clocked computation
simulation constructs the circuit in time polynomial in the clock bound,
so its description length satisfies
\begin{equation}
  S(n)=|D_{x,r}|\leq2^{O(\Llog(n)^d)}.
 \label{eq:compiled-size}
\end{equation}
On the event in \eqref{eq:common-tape},
\[
 D_{x,r}\in\CSAT\iff\exists u\ B(x,u)=1\iff x\in L.
\]
Apply the assumed circuit-satisfiability algorithm to $D_{x,r}$, using
fresh coins and reducing its error to at most $1/64$. Its time is
\begin{align}
 2^{O(\Llog(S(n))^d)}
 &=2^{O((\Llog(n)^d)^d)}\notag\\
 &=2^{O(\Llog(n)^{d^2})}.
 \label{eq:time-compose}
\end{align}
Construction of $D_{x,r}$ is within this bound. By a union bound over
failure of \eqref{eq:common-tape} and the final circuit-satisfiability
test, total error is at most $1/32$. Complementation gives the same
bound for $\Pi_2^p$.
\end{proof}

\begin{remark}[The explicit-description cost]
The input to the final SAT algorithm in this proof has quasipolynomial
length $S(n)$, not length $n$. Formula \eqref{eq:time-compose} records
this composition. The proof therefore yields exponent $d^2$, not $d$,
and does not turn a quasipolynomial circuit into a polynomial-size one.
\end{remark}

\section{A level-independent simulation of the polynomial hierarchy}
\label{sec:ph}

The following two results are imported from the literature. They are not
consequences attributed to the binary-certificate theorem.

\begin{theorem}[Polynomial-arity advice; {\cite[Theorem~4.4, long version]{ABG03}}]
\label{thm:abg}
Let $k(n)$ be polynomially bounded and let $A$ be a language. Suppose a
relation $W(t,b)$, on $k(n)$-tuples of $n$-bit strings and $k(n)$-bit
vectors, satisfies
\[
 \forall t\ \exists b\ W(t,b),
 \qquad W(t,b)\Longrightarrow b\neq\chi_A(t).
\]
Then
\[
 A\in\NP^W/\poly\ \cap\ \mathrm{coNP}^W/\poly.
\]
In particular,
\begin{equation}
 \poly\mc\subseteq\NP/\poly\cap\coNP/\poly.
 \label{eq:abg}
\end{equation}
\end{theorem}

To obtain \eqref{eq:abg}, take $W(t,b)$ to mean that the
polynomial-time comparator outputs $b$ on $t$. Then $W\in\Pclass$.
This uses the original theorem's two-sided nondeterministic advice
conclusion, rather than the weaker summary in
\cite[Section~3.1]{BD13}.

\begin{theorem}[Strengthened Yap collapse; \cite{CCHO05}]
\label{thm:yap}
If $\NP\subseteq\coNP/\poly$, then
\begin{equation}
 \PH=\Sclass^{\NP}\subseteq\ZPP^{\Sigma_2^p}.
 \label{eq:yap}
\end{equation}
The last containment uses the relativized symmetric-alternation
simulation $\Sclass\subseteq\ZPP^{\NP}$ of Cai~\cite{cai-s2zpp}, as also recorded in~\cite{CCHO05}.
\end{theorem}

Here $\Sclass^\NP$ is symmetric alternation with a polynomial-time
NP-oracle verifier: each yes-instance has a polynomial-length certificate
that wins against all opposing certificates, and each no-instance has a
polynomial-length certificate that defeats all opposing certificates.
We also use the standard containment
$\Sclass^{\NP}\subseteq\Sigma_3^p\cap\Pi_3^p$. Thus
\eqref{eq:yap} supplies a third-level polynomial-time characterization
as well as a zero-error oracle simulation. For $\ZPP^{\Sigma_2^p}$,
the expected polynomial time and zero-error correctness are with the
exact oracle.

\begin{proposition}
\label[proposition]{prop:structural}
Under $\hyp$, the known structural collapse
\[
 \PH=\Sclass^{\NP}\subseteq\ZPP^{\Sigma_2^p}
\]
holds.
\end{proposition}

\begin{proof}
The arity in $\hyp$ is polynomially bounded. By \cref{thm:abg},
$\SAT\in\coNP/\poly$. Polynomial-time reductions to SAT, with advice
for all possible reduction-output lengths bundled together, give
$\NP\subseteq\coNP/\poly$. Equivalently,
$\coNP\subseteq\NP/\poly$. Apply \cref{thm:yap}.
\end{proof}

\begin{theorem}[Uniform randomized simulation of all of $\PH$]
\label{thm:ph}
Under $\hyp$,
\[
 \PH\subseteq\BPTIME\bigl(2^{O(\Llog(n)^{d^2})}\bigr).
\]
\end{theorem}

\begin{proof}
Fix $L\in\PH$. By \cref{prop:structural}, there are an exact oracle
$A\in\Sigma_2^p$ and a zero-error randomized oracle machine $M^A$
whose expected running time is at most a polynomial $p(n)$.
Truncate the computation after $64p(n)$ steps, returning an arbitrary
bit on timeout. With the exact oracle, Markov's inequality bounds the
timeout probability by $1/64$. Set $Q(n)=\lceil64p(n)\rceil+1$.
Every truncated path makes at most $Q(n)$ queries, each of length at
most $Q(n)$.

By \cref{cor:random-sat,lem:second}, $A$ has a bounded-error algorithm
with time $2^{O(\Llog(t)^{d^2})}$ on inputs of length $t$. Amplify it to
error at most $1/(64Q(n))$ on each query. The additional factor is
$O(\log Q(n))$.

Run the truncated machine, replacing each oracle query by this algorithm
with fresh random coins. To handle adaptive queries, couple this run to
the truncated exact-oracle run using the same outer random tape. Until
the first incorrect simulated answer, the two runs make identical
queries. Conditional on all preceding query answers being correct, the
new query is fixed by the history and the fresh simulation coins give
error at most $1/(64Q(n))$. Summing the probabilities of the possible
first errors bounds any divergence by $1/64$.

Outside the timeout and divergence events, the answer is correct.
Total error is therefore at most $1/32$. There are polynomially many
simulated queries of polynomial length. Their combined, worst-case
clocked running time is
\[
 \poly(n)\,
 2^{O(\Llog(Q(n))^{d^2})}
 =2^{O(\Llog(n)^{d^2})}.
\]
The polynomial $p$ and hidden constants may depend on $L$, but the
exponent $d^2$ does not depend on its level in $\PH$.
\end{proof}

\section{The quasipolynomial hierarchy}
\label{sec:qh}

\begin{lemma}[Closure under a quasipolynomial quantifier block]
\label[lemma]{lem:qp-block}
Assume $\CSAT\in\RTIME(2^{O(\Llog(n)^d)})$ for fixed $d$.
Let $B(x,u)$ have a bounded-error randomized quasipolynomial-time
algorithm, and let $q(n)$ be a quasipolynomial, time-constructible bound.
Then the languages
\[
 \{x:\exists u\in\bits^{q(|x|)}\ B(x,u)
                         \},\qquad
 \{x:\forall u\in\bits^{q(|x|)}\ B(x,u)
                         \}
\]
belong to $\BPQP$.
\end{lemma}

\begin{proof}
For the existential language, amplify the algorithm for $B$ to error
$2^{-q(n)-6}$ for every $u$. The repetition count is $O(q(n)+1)$,
not $2^{q(n)}$. Both the combined input length $n+q(n)$ and the
amplified running time are quasipolynomial in $n$: fixed compositions
and products of quasipolynomials remain quasipolynomial.

Sample a common random tape for all $u$ and use the union bound as in
\eqref{eq:common-tape}. Compile the fixed-tape computation into a circuit
with $q(n)$ free inputs and quasipolynomial explicit description length.
The assumed circuit-satisfiability algorithm decides its satisfiability
in quasipolynomial time. The same two-event error analysis applies.
For the universal language, complement $B$, use the existential result,
and complement the final answer.
\end{proof}

\begin{theorem}[Collapse at quasipolynomial time]
\label{thm:qh}
Under $\hyp$, $\QH=\BPQP$.
\end{theorem}

\begin{proof}
For $\QH\subseteq\BPQP$, express any fixed constituent of $\QH$ by a
fixed number of alternating quantified blocks of quasipolynomial length
and a deterministic quasipolynomial-time predicate. This form also follows
by encoding the choices in each phase of a bounded-alternation machine;
only choices actually used by a branch affect the deterministic
simulation. Starting with the innermost block, apply
\cref{lem:qp-block} repeatedly. The number of applications is fixed, so
the resulting running time remains quasipolynomial. This direction uses
$\hyp$ only through \cref{cor:random-sat}.

For the reverse inclusion, which is unconditional, let a bounded-error
machine have time bound $t(n)=2^{O(\Llog(n)^c)}$, with $t(n)\geq n+2$.
Repeat it $K\lceil\log(t(n)+2)\rceil$ times and take the majority.
For a sufficiently large fixed $K$, Chernoff amplification yields error
$\varepsilon(n)$ and a padded coin count $m(n)$ satisfying
\[
 m(n)=O(t(n)\log(t(n)+2)),
 \qquad \varepsilon(n)<\frac1{4m(n)}.
\]
Indeed, the error decreases as $t(n)^{-\Omega(K)}$, whereas $m(n)$
increases by only a logarithmic factor. This establishes the inequality
without defining the repetition count circularly in terms of the final
number of coins. Increase the constant for finitely many small inputs.
Write the amplified, deterministic-on-its-tape computation as $D(x,r)$
for $r\in\bits^m$.

Fix $x$, and let $G_x=\{r:D(x,r)=1\}$. For a yes-instance,
$|G_x|\geq(1-\varepsilon)2^m$. Choose $m$ independent uniform shifts
$s_1,\ldots,s_m\in\bits^m$. Any fixed point is uncovered by
$\bigcup_i(G_x\oplus s_i)$ with probability at most $\varepsilon^m$.
A union bound bounds the probability that some point is uncovered by
$2^m\varepsilon^m<1$. Thus there exists a collection of $m$ shifts
covering the whole cube.

For a no-instance, $|G_x|\leq\varepsilon2^m$. For every collection of
$m$ shifts, their union has size at most
$m\varepsilon2^m<2^m$, so it cannot cover. Consequently,
\begin{equation}
 x\in L\iff
 \exists s_1,\ldots,s_m\in\bits^m\
 \forall r\in\bits^m\quad
 \bigvee_{i=1}^m D(x,r\oplus s_i)=1.
 \label{eq:cover}
\end{equation}
The existential block has length $m^2$, the universal block has length
$m$, and evaluating the disjunction takes quasipolynomial time. This
is a second-level quasipolynomial alternating computation, proving
$\BPQP\subseteq\QH$.
\end{proof}

\begin{remark}
The proof eliminates a fixed number of alternating blocks with
quasipolynomial, rather than polynomial, resource bounds. The covering
argument \eqref{eq:cover} accounts for all its shifts as quantified
strings; none are free advice.
\end{remark}

\section{Exponential-time padding consequences}
\label{sec:exp}

\begin{theorem}
\label{thm:exp}
Under $\hyp$,
\[
 \NEXP=\REXP,\qquad \UEXP=\EXP.
\]
\end{theorem}

\begin{proof}
Let $L\in\NEXP$. Fix a nondeterministic machine for $L$ with clock
$t(n)=2^{n^a}$ for some constant $a$. On input $x$, a standard explicit
computation encoding constructs, in $t(n)^{O(1)}$ time, a circuit or
formula $F_x$ such that
\[
 x\in L\iff F_x\in\SAT,
 \qquad |F_x|=N\leq t(n)^{O(1)}=2^{n^{O(1)}}.
\]
Run \cref{cor:random-sat} on $F_x$. Construction and solution together
take $2^{n^{O(1)}}$ time and preserve one-sided error, since an
unsatisfiable $F_x$ is never accepted. Hence $\NEXP\subseteq\REXP$.
Conversely, guess the random tape of an $\REXP$ algorithm. Its uniform
clock bounds the tape length exponentially; one-sided soundness excludes
accepting tapes on no-instances, while a yes-instance has an accepting
tape. Thus $\REXP\subseteq\NEXP$.

Now suppose the original machine is unambiguous. Choose a canonical
fixed-length encoding of its nondeterministic choices: a deterministic
verifier consumes a choice bit only at a binary nondeterministic step
and requires all unused bits of the encoding to be zero. Halting and
padding conventions are fixed. Each computation branch then has exactly
one choice encoding. Compile this verifier into a circuit. Its internal
gate values are determined by the free choice bits; a subsequent gate
encoding as a formula likewise gives each accepting input a unique
extension. The resulting circuit or formula has exactly as many
satisfying assignments as the machine has accepting computations, hence
zero or one. Apply \cref{prop:unique}. The time is
\[
 2^{O(\Llog(N)^d)}=2^{n^{O(1)}},
\]
so $\UEXP\subseteq\EXP$. Deterministic machines are unambiguous,
which supplies the reverse inclusion.
\end{proof}

\subsection{The exponential hierarchy}

For fixed $k\ge1$, let $\Sigma_k\mathrm{EXP}$ be the union over $c\ge1$
of the classes of languages having a characterization by $k$ alternating
quantifier blocks, beginning existentially, over strings of length at
most $2^{n^c}$, with a deterministic predicate clocked in time
$2^{n^c}$. Put $\EXPH=\bigcup_k\Sigma_k\mathrm{EXP}$ and
$\BPEXP=\bigcup_c\BPTIME(2^{n^c})$. For a
time-constructible $t(n)\ge n+2$, write
\[
  \mathrm{pad}_t(x)=x\,1\,0^{t(|x|)}.
\]
This map is injective and computable in time $O(t(|x|))$.
Its delimiter is the rightmost one in the padded string. To recognize
its image on a string of length $N$, find that delimiter, recover the
candidate prefix $x$, and compare the trailing-zero count with $t(|x|)$.
Use a fixed time constructor for $t$ and truncate it after a sufficiently
large constant times $N$ steps. If $t(|x|)\le N$, the constructor finishes;
a timeout therefore cannot reject a valid padded string. A reported value
larger than $N$, or one unequal to the trailing-zero count, is rejected.
Thus image recognition and recovery are polynomial-time even on malformed
inputs; they never run for an unbounded value of $t(|x|)$.

\begin{lemma}[Padding a bounded-alternation characterization]
\label[lemma]{lem:padding}
Assume $\hyp$. Let $L$ have a characterization by a fixed number $k$ of
alternating quantifier blocks, beginning existentially, over strings
of length at most $t(n)$,
with a deterministic predicate clocked in time $t(n)^{O(1)}$, where
$t(n)\ge n+2$ is time-constructible. Then $L$ is decided with error at
most $1/3$ in time
\[
  t(n)^{O(1)}+2^{O(\Llog(t(n))^{d^2})}.
\]
\end{lemma}

\begin{proof}
Let $L'=\{\mathrm{pad}_t(x):x\in L\}$. On a string of length $N$, a
polynomial-time predicate first checks that the string lies in the
image of $\mathrm{pad}_t$ and recovers $x$, so that $t(|x|)<N$. It then
evaluates the given predicate of $L$. The original blocks can be padded to a common length at most a
polynomial in $N$ and interpreted canonically, ignoring unused bits.
Their original length bounds are reconstructed from $x$. The combined
predicate is clocked in time $N^{O(1)}$. Strings outside the image are rejected by this predicate
under every choice of the quantified strings. Hence
$L'\in\Sigma_k^p\subseteq\PH$, and \cref{thm:ph} gives a
bounded-error algorithm for $L'$ with time $2^{O(\Llog(N)^{d^2})}$.
To decide $x$, construct $\mathrm{pad}_t(x)$ and run that algorithm.
Since $N=|x|+1+t(|x|)\le 2t(|x|)$, we have
$\Llog(N)=O(\Llog(t(|x|)))$.
\end{proof}

\begin{corollary}\label[corollary]{cor:exph}
Under $\hyp$, $\EXPH=\BPEXP$.
\end{corollary}

\begin{proof}
Let $L\in\Sigma_k\mathrm{EXP}$ with bound $t(n)=2^{n^c}$,
modified on finitely many short lengths so that $t(n)\ge n+2$.
Lemma~\ref{lem:padding} gives time
$2^{O(n^c)}+2^{O(n^{cd^2})}$, so $L\in\BPEXP$.

For the reverse inclusion, which is unconditional, apply the covering
argument from the proof of \cref{thm:qh} with clock $t(n)=2^{n^c}$. The
amplified coin count is $m(n)=2^{O(n^c)}$, the existential block has
$m(n)^2$ bits, and the universal block and the predicate clock are
$2^{O(n^c)}$. Thus $\BPEXP\subseteq\Sigma_2\mathrm{EXP}$.
\end{proof}

\subsection{Unconditional diagonalization and conditional randomized simulation}
\label{sec:nd-hard-extension}

We first construct the hard languages by diagonalization and prefix
projection, then apply the randomized simulation of $\hyp$.

Circuits have fan-in-two gates over the full binary basis, which contains
all sixteen two-input functions, including constants, and one designated
output gate. Size counts gates. Descriptions use a fixed format recording
each gate's function and two predecessors; their length is
$O(s\log(m+s))$ for $m$ ordinary inputs and $s$ gates. Invalid descriptions
are evaluated by a fixed default convention and guarded where required.
Unused gates are permitted, and the output need not be the last gate.

An existential circuit is a Boolean circuit $D(x,z)$ interpreted by
$x\mapsto\exists z\,D(x,z)$. A circuit with $s$ gates uses at most $2s$
distinct witness inputs. Renumbering them and allowing unused inputs,
we may take $z\in\bits^{2s}$. A \emph{signed existential circuit} is a
pair $(D,\sigma)$, $\sigma\in\bits$, computing
\[
 f_{D,\sigma}(x)=\bigl(\exists z\,D(x,z)\bigr)\oplus\sigma.
\]
Polynomial-size families with sign fixed to $0$ characterize
$\NP/\poly$; those with sign fixed to $1$ characterize $\coNP/\poly$.
The lower bounds below exclude both signs separately at each stated
input length. Write
$\Sigma_5 E=\Sigma_5\mathrm{TIME}(2^{O(n)})$ for five-block alternating
single-exponential time, not for a polynomial-time oracle class.

\begin{lemma}[Unconditional signed nondeterministic diagonalization]
\label[lemma]{lem:unconditional-nd-hard}
There is a language $H\in\Sigma_5 E$ whose
length-$m$ slice, for every $m\ge6$, is computed by no signed existential
circuit with at most
\[
 s(m)=\left\lfloor\frac{2^m}{10m}\right\rfloor
\]
gates.
\end{lemma}

\begin{proof}
Fix $m\ge6$ and $s=s(m)\ge1$. Circuits with fewer than $s$ gates can
be padded by unused gates. With $m$ ordinary inputs, $2s$ witness inputs,
and $s$ gate positions, there are at most
\[
 s\bigl(16(m+3s)^2\bigr)^s
\]
underlying circuits, and at most twice that many signed functions.
Counting even cyclic or otherwise invalid predecessor choices only
increases this upper bound. Since $m+3s\le2^{m+2}$, its base-two logarithm
is at most
\[
 m+1+s(2m+8)
 \le m+1+2^m\left(\frac15+\frac{4}{5m}\right)
 <2^m.
\]
For the last inequality, $(m+1)/2^m\le7/64$ and
$1/5+4/(5m)\le1/3$ for $m\ge6$. Some truth table on $m$ bits therefore
has neither type of circuit. Let $T_m$ be the lexicographically first
such table, and define $H(x)=T_{|x|}[x]$ for $|x|\ge6$, with $H(x)=0$
on shorter inputs.

We exhibit a fixed fifth-level characterization with exponential resource
bounds. Let $\mathcal D_m$ denote the well-formed circuit descriptions
with $s$ gate slots, $m$ ordinary inputs, and $2s$ witness inputs. All
such descriptions have one fixed length $2^{O(m)}$. For a truth table
$T$, put $t_{T,\sigma}(y)=T[y]\oplus\sigma$ and define
\begin{align*}
 M(T,D,\sigma;y,a,b)
 &=\bigl(t_{T,\sigma}(y)\land\neg D(y,b)\bigr)
   \lor\bigl(\neg t_{T,\sigma}(y)\land D(y,a)\bigr),\\
 E(T,D,\sigma;y,a,b)
 &=\bigl(t_{T,\sigma}(y)\land D(y,b)\bigr)
   \lor\bigl(\neg t_{T,\sigma}(y)\land\neg D(y,a)\bigr).
\end{align*}
Here $y$ has $m$ bits and $a,b$ have $2s$ bits. Directly from the two
possible values of $t_{T,\sigma}(y)$,
\begin{align*}
 T\ne f_{D,\sigma}
 &\iff\exists(y,a)\,\forall b\ M(T,D,\sigma;y,a,b),\\
 T=f_{D,\sigma}
 &\iff\forall(y,a)\,\exists b\ E(T,D,\sigma;y,a,b).
\end{align*}
Thus $T$ is hard exactly when
$\forall(D,\sigma)\exists(y,a)\forall b\ M$, with an invalid-description
guard. An earlier table $T'$ has a small signed circuit exactly when
$\exists(D',\sigma')\forall(y',a')\exists b'\ E$, with a validity test.

For $x\in\bits^m$, the characterization of $H(x)=1$ is therefore
\[
 \exists T\ \forall(D,\sigma,T')\
 \exists(y,a,D',\sigma')\ \forall(b,y',a')\ \exists b'\ \Phi,
\]
where $\Phi$ is the conjunction of
\begin{gather*}
 T[x]=1,\\
 D\notin\mathcal D_m\ \lor\ M(T,D,\sigma;y,a,b),\\
 T'\not<_{\mathrm{lex}}T\ \lor\
 \bigl(D'\in\mathcal D_m\land
 E(T',D',\sigma';y',a',b')\bigr).
\end{gather*}
Tables $T,T'$ have $2^m$ bits. Invalid circuit descriptions are evaluated
by a fixed default convention and are handled by the displayed guards.
The hardness and minimality clauses use disjoint subsidiary variables;
merging their blocks preserves their meanings. The formula therefore
selects exactly $T=T_m$ and its bit at $x$.

Every block and the matrix computation have size $2^{O(m)}$, even if
truth-table access is implemented by a full scan. Hence
$H\in\Sigma_5\mathrm{TIME}(2^{O(m)})=\Sigma_5 E$ unconditionally.
The circuit lower bounds hold at every $m\ge6$ by the definition of $T_m$.
\end{proof}

\begin{lemma}[Unconditional slowly growing prefix projection]
\label[lemma]{lem:unconditional-projection}
Let $h(n)\ge1$ be a polynomial-time computable integer-valued function
with $h(n)\to\infty$, and suppose
$\ell(n)=\Llog(n)h(n)\le n$ for all sufficiently large $n$.
There is a language
\[
 G_h\in\Sigma_5\mathrm{TIME}\bigl(2^{O(\ell(n))}\bigr)
 \setminus\bigl(\NP/\poly\cup\coNP/\poly\bigr),
\]
without assuming $\hyp$. Its length-$n$ slice has no signed existential
circuit with at most
\[
 \left\lfloor\frac{2^{\ell(n)}}{10\ell(n)}\right\rfloor
\]
gates for all sufficiently large $n$.
\end{lemma}

\begin{proof}
Fix a cutoff after which $6\le\ell(n)\le n$, and use the unconditional
language $H$ from \cref{lem:unconditional-nd-hard}. Above the cutoff set
\[
 G_h(x)=H(x_1\cdots x_{\ell(|x|)}),
\]
and set it to zero below the cutoff. Substitute this prefix into the
five-block characterization of $H$. All blocks and the matrix have size
$2^{O(\ell(n))}$. Polynomial-time computation of $h$ and prefix extraction
are absorbed in that bound, since $\ell(n)\ge\Llog(n)$ eventually.
This proves the alternating-time upper bound unconditionally.

If a signed existential circuit of the indicated size computed $G_h$
on length-$n$ inputs, fix its final $n-\ell(n)$ ordinary input bits to
zero. With the full binary basis, replacing constant predecessors does
not increase gate count: the resulting Boolean function at each gate is
again one of the sixteen available functions, using a surviving input
as a dummy predecessor if necessary. The sign and witness inputs are
unchanged. The restricted circuit computes $H$ on every $\ell(n)$-bit
input, contradicting \cref{lem:unconditional-nd-hard}.

For every fixed $k$,
\[
 \log_2\left(\frac{2^{\ell(n)}}{10\ell(n)n^k}\right)
 =\ell(n)-k\log_2 n-O(\log\ell(n))\longrightarrow+\infty.
\]
Indeed, $\ell(n)/\log_2 n\ge h(n)\to\infty$ and
$\log\ell(n)\le\log n$ eventually. Thus the lower bound exceeds every
constant multiple of every fixed power of $n$. The language $G_h$ does
not depend on a proposed advice exponent $k$.
\end{proof}

\begin{corollary}[Unconditional hardness in the quasipolynomial hierarchy]
\label[corollary]{cor:unconditional-qh-hard}
Unconditionally,
\[
 \QH\not\subseteq\NP/\poly\cup\coNP/\poly.
\]
In particular, the choice
$h(n)=\max\{1,\lceil\log_2\Llog(n)\rceil\}$ gives a single language
\[
 G_h\in\Sigma_5\mathrm{TIME}\!\left(
          2^{O(\Llog(n)\log\Llog(n))}\right)
 \subseteq\QH
\]
that lies outside both advice classes.
\end{corollary}

\begin{proof}
Apply \cref{lem:unconditional-projection}. For this choice,
$\ell(n)=O(\Llog(n)\log\Llog(n))$ is polylogarithmic, so the displayed
alternating-time class is a constituent of $\QH$.
For a general $h$ in that lemma, membership in $\QH$ is asserted only
when $\ell(n)\le\Llog(n)^{O(1)}$; the condition $\ell(n)\le n$ alone
does not give this containment.
\end{proof}

\begin{theorem}[Quantitative randomized simulation of the diagonal languages]
\label{thm:diagonal-randomized}
Assume $\hyp$. The same unconditionally defined languages satisfy
\[
 H\in\BPTIME\bigl(2^{O(m^{d^2})}\bigr),\qquad
 G_h\in\BPTIME\bigl(2^{O((\Llog(n)h(n))^{d^2})}\bigr).
\]
In particular,
\[
 \BPEXP\not\subseteq\NP/\poly\cup\coNP/\poly.
\]
\end{theorem}

\begin{proof}
Apply \cref{lem:padding} to the unconditional five-block characterization
of $H$, with $t(m)=2^{c(m+1)}$ for a sufficiently large fixed integer $c$.
The resulting clock is
$2^{O(m)}+2^{O(m^{d^2})}=2^{O(m^{d^2})}$.
For $G_h$, run this algorithm on the prefix of length $\ell(n)$.
Polynomial preprocessing is absorbed by the stated bound. The algorithm
is uniform and has error at most $1/3$ on every input. The near-maximum
hardness of $H$ eventually exceeds every polynomial, giving the final
assertion.
\end{proof}

A deterministic circuit is an existential circuit with no witness
inputs and sign $0$, of the same gate size. Thus the deterministic
exponential circuit lower bound is included in the same theorem.

\begin{corollary}[Arbitrarily small exponent overhead]
\label[corollary]{cor:qp-no-np-poly}
Under $\hyp$, the fixed language from \cref{cor:unconditional-qh-hard}
belongs to
\[
 \BPTIME\left(2^{O((\Llog(n)\log\Llog(n))^{d^2})}\right)
 \setminus\bigl(\NP/\poly\cup\coNP/\poly\bigr).
\]
For every fixed $\varepsilon>0$, the same language also belongs to
\[
 \BPTIME\left(2^{O(\Llog(n)^{d^2+\varepsilon})}\right)
 \setminus\bigl(\NP/\poly\cup\coNP/\poly\bigr).
\]
\end{corollary}

\begin{proof}
Use \cref{thm:diagonal-randomized} with
$h(n)=\max\{1,\lceil\log_2\Llog(n)\rceil\}$ for the first assertion.
For each fixed $\varepsilon>0$,
\[
 (\log\Llog(n))^{d^2}=o(\Llog(n)^\varepsilon),
\]
so that very same language has the second displayed time bound.
\end{proof}

\begin{remark}[Exact-clock fixed-polynomial lower bounds]
\label[remark]{rem:kannan-clock}
For each fixed $k\ge1$, projecting $H$ to
$\ell_k(n)=(k+2)\Llog(n)$ ordinary input bits, with a fixed convention at
short lengths, gives $L_k\in\Sigma_5^p$ whose circuit size eventually
exceeds every constant multiple of $n^k$; under $\hyp$,
\cref{thm:ph} places $L_k$ in
$\BPTIME(2^{O(\Llog(n)^{d^2})})$.
This retains the exact exponent $d^2$ but permits a different $L_k$ for
each $k$, unlike the single-language conclusion of
\cref{cor:qp-no-np-poly}.
\end{remark}

\begin{corollary}[Conditional bounded-error time separation without advice]
\label[corollary]{cor:conditional-bpp-hierarchy}
Under $\hyp$, for every fixed $\varepsilon>0$,
\begin{equation}
 \mathrm{BPP}\subsetneq
 \BPTIME\!\left(2^{O(\Llog(n)^{d^2+\varepsilon})}\right).
 \label{eq:conditional-bpp-hierarchy}
\end{equation}
More precisely,
\[
 \mathrm{BPP}\subsetneq
 \BPTIME\!\left(2^{O((\Llog(n)\log\Llog(n))^{d^2})}\right).
\]
Separately, each displayed larger time class is contained in $\BPQP$.
The strict inclusions compare uniformly decided languages and use no advice.
\end{corollary}

\begin{proof}
Adleman's inclusion gives
$\mathrm{BPP}\subseteq\Pclass/\poly$~\cite{Adl78}. For completeness,
amplify a BPP algorithm's error below $2^{-n-1}$ and fix, by a union
bound over $\bits^n$, one polynomial-length random tape correct on all
length-$n$ inputs. Hardwiring this tape yields polynomial-size circuits.
The unconditional $G_h$ from \cref{cor:unconditional-qh-hard} is outside
$\NP/\poly\cup\coNP/\poly$, hence outside $\Pclass/\poly$ and BPP.
Under $\hyp$, \cref{cor:qp-no-np-poly} places that language in the
logarithmic-overhead clock and in the displayed class for every
$\varepsilon>0$. Each of these clocks eventually dominates every
polynomial, proving containment of BPP as well as strictness.

The qualitative dependency chain is
\[
 \hyp\xRightarrow{\text{\cref{cor:random-sat,thm:qh}}}
 \QH\subseteq\BPQP,
 \qquad
 G_h\in\QH\setminus\Pclass/\poly\quad\text{unconditionally}.
\]
The numerical exponent in \eqref{eq:conditional-bpp-hierarchy} uses,
in addition, the explicit PH bound in \cref{thm:ph}, retained by
\cref{lem:padding,thm:diagonal-randomized}. It is not obtained merely
by substituting into the unparameterized equality $\QH=\BPQP$.
\end{proof}

\paragraph{Position relative to known separations.}
Karpinski--Verbeek proved, by padding and deterministic time hierarchies,
that
\[
 \mathrm{BPP}\subsetneq\class{BPSUBEXP},\qquad
 \class{BPSUBEXP}=\bigcap_{\eta>0}\BPTIME(2^{n^\eta})
\]
\cite{karpinski-verbeek87}; their chapter explicitly gives a separation
from the fixed clock $n^{\log n}$ as well. Lu, Oliveira, and Santhanam
frame the remaining smaller-gap question as whether
$\BPTIME(n)\subsetneq\BPTIME(T(n))$ for a clock $T$ that remains
subexponential after any fixed number of self-compositions
\cite[Section~1]{lu-oliveira-santhanam21}.
For a fixed $a>1$ and $T_a(n)=2^{C\Llog(n)^a}$, every fixed iterate
satisfies
\[
 T_a^{\circ r}(n)=2^{O(\Llog(n)^{a^r})}=2^{o(n)}.
\]
Thus \eqref{eq:conditional-bpp-hierarchy}, with $a=d^2+\varepsilon$,
is an advice-free conditional separation in precisely this
self-composition-stable regime, rather than merely a separation from
BPP to an unspecified larger randomized time class.

There are distinct conditional and advice-based routes to hierarchy
results. Barak proved a hierarchy with $O(\log\log n)$ advice bits and
also observed that fully uniform machines have a hierarchy if BPP has a
complete problem~\cite{Bar02}. Fortnow and Santhanam reduced the advice
to one bit~\cite{FS04}. Lu, Oliveira, and Santhanam give further
conditional hierarchies from improvements to pseudodeterministic
algorithms for circuit acceptance probabilities
\cite[Theorem~2]{lu-oliveira-santhanam21}. The promise-BPTIME hierarchy
of~\cite{He25} concerns promise problems rather than total languages.

The qualitative implication
\[
 \QH\subseteq\BPQP\quad\Longrightarrow\quad
 \mathrm{BPP}\subsetneq\BPQP
\]
uses the unconditional diagonal language and Adleman's theorem.
For example, $\NP\subseteq\mathrm{BPP}$ supplies a polynomial-time
randomized SAT algorithm, and \cref{lem:qp-block,thm:qh} then give
$\QH\subseteq\BPQP$. Under $\hyp$, the parameter-preserving simulation
supplies the explicit clock in \eqref{eq:conditional-bpp-hierarchy}.
The oracle in \cref{cor:oracle-hd} complements this upper bound:
relative to it the same hypothesis holds, but NP itself has no
bounded-error algorithm of clock $2^{\delta\Llog(n)^d}$ for
$0<\delta<c_0$.

\begin{corollary}[The advice containment separates PH from QH]
\label[corollary]{cor:ph-strict-qh}
If $\coNP\subseteq\NP/\poly$, then
\begin{gather*}
 \PH/\poly=\NP/\poly=\coNP/\poly,\\
 \QH\not\subseteq\PH/\poly,\qquad \PH\subsetneq\QH.
\end{gather*}
In particular these conclusions hold under $\hyp$. Under $\hyp$,
\cref{thm:qh} additionally gives
\[
 \PH\subsetneq\QH=\BPQP,\qquad
 \BPQP\not\subseteq\PH/\poly.
\]
\end{corollary}

\begin{proof}
By \cref{lem:ph-np-advice}, the advice hypothesis implies
$\PH\subseteq\NP/\poly$ and $\NP/\poly=\coNP/\poly$.
Applying an additional polynomial advice string to an NP/poly language
adds no power: bundle the outer advice with the inner advice at every
relevant polynomially bounded encoding length. Thus
$\PH/\poly\subseteq\NP/\poly$; the reverse inclusion follows from
$\NP\subseteq\PH$.
The unconditional language in \cref{cor:unconditional-qh-hard} belongs
to QH but not to this advice class. Since $\PH\subseteq\QH$ always,
this proves the strict inclusion without using \cref{thm:ph} or any
randomized simulation. Under $\hyp$, \cref{thm:abg} supplies the
advice hypothesis; combining the conclusions with \cref{thm:qh} gives
the two statements involving BPQP.
\end{proof}

\begin{remark}[Other semantic classes]
Unconditional circuit lower bounds include
$\class{MA}_{\EXP}\not\subseteq\Pclass/\poly$~\cite{bft98} and
near-maximum deterministic circuit hardness for symmetric
single-exponential time: Chen, Hirahara, and Ren proved a version with
one advice bit~\cite{chen-hirahara-ren24}, and Li removed that
advice~\cite{li24}. The present randomized placement instead uses the
PH simulation and retains its explicit longer clock.
\end{remark}

Together, \cref{prop:unique,cor:random-sat,thm:ph,thm:qh,thm:exp,cor:exph}
prove the unique-solution, randomized, and larger-time-scale assertions.
The unconditional lemmas in this subsection supply the diagonal languages;
\cref{thm:diagonal-randomized,cor:qp-no-np-poly,rem:kannan-clock,cor:conditional-bpp-hierarchy}
give their conditional randomized placements and time separations.
\Cref{cor:ph-strict-qh} separates the independent advice argument from
these simulations. The deterministic-advice assertion is proved in
\cref{cor:ph-advice}, and the symmetric-alternation assertion in
\cref{thm:s2-simulation}.

\section{Advice-preserving consequences}
\label{sec:advice}

The uniform simulation of $\PH$ above uses no designated advice in its
final algorithm. The binary-certificate theorem in~\cite{BD26} gives a
different consequence, where the small advice bound is retained.

\begin{proposition}[Time-scaled advice-preserving isolation]
\label[proposition]{prop:advice}
Under $\hyp$, suppose a polynomial-time predicate $R$, polynomially
bounded polynomial-time computable functions $a,w$, and a designated advice sequence
$\alpha_n\in\bits^{a(n)}$ satisfy
\[
 x\in L\iff\exists z\in\bits^{w(n)}\ R(x,\alpha_n,z)=1.
\]
Then
\[
 L\in\DTIME\bigl(2^{O(\Llog(n)^d)}\bigr)
            /O(a(n)+w(n)+n).
\]
Consequently,
\begin{equation}
 2\mc\subseteq
 \DTIME\bigl(2^{O(\Llog(n)^d)}\bigr)/O(n).
 \label{eq:binary-advice}
\end{equation}
\end{proposition}

\begin{proof}
If $w(n)=0$, evaluate $R$ directly with the designated advice at that
length. Otherwise, compile $R(x,\alpha_n,\cdot)$ into a circuit on exactly $w(n)$ free
input bits. Use \cref{prop:unique} in the isolation procedure of
\cref{lem:isolation}. For each positive $n$-bit input, seeds making
that procedure accept occupy at least $3/16$ of the $3w(n)+1$-bit seed
space; every seed rejects a negative input.

By the explicit-expander hitting lemma in
\cite[Lemma~F.5]{BD26}, one walk of length $O(n)$ hits all these at
most $2^n$ dense good-seed sets. Its description uses
$3w(n)+1+O(n)$ bits, and its vertices are reconstructed in polynomial
time. Append this description to $\alpha_n$. The deterministic
interpreter reconstructs the walk and ORs the isolation tests at its
vertices. This is correct simultaneously on all inputs of length $n$.
All circuits have polynomial size and there are polynomially many
promise-solver calls, so the time remains
$2^{O(\Llog(n)^d)}$. This is exactly the resource accounting in
\cite[Theorem~7.1]{BD26}, with only the solver time changed.

For \eqref{eq:binary-advice}, use the verifier from
\cite[Theorem~3.1]{BD26}: its common advice has $3n+5$ bits and its
certificates have at most $5n+12$ bits. A fixed-length version of that
verifier, provided in the source, has the same existential semantics.
Substitution makes the retained advice $O(n)$.
\end{proof}

\begin{remark}[No free seed and no robust-advice claim]
The walk in \cref{prop:advice} is counted as ordinary advice, fixed for
all inputs of a length. In \cref{cor:random-sat} the seeds are instead
sampled at runtime. These are different conclusions. Neither conclusion
makes arbitrary candidate advice sound or complete. The designated-advice
qualification of~\cite[Corollary~3.5 and Appendix~H]{BD26} is preserved.
\end{remark}

\subsection{A deterministic, level-independent bound for all of $\PH$}

\begin{lemma}[Closure of nondeterministic polynomial advice]
\label[lemma]{lem:ph-np-advice}
If $\coNP\subseteq\NP/\poly$, then
\[
 \NP/\poly=\coNP/\poly,
 \qquad
 \PH\subseteq\NP/\poly\cap\coNP/\poly.
\]
The intersection here is of two advice classes; it does not assert a
robust $(\NP\cap\coNP)$ predicate on arbitrary advice--input pairs.
\end{lemma}

\begin{proof}
Put $\mathcal D=\NP/\poly$. For $L\in\mathcal D$, choose $K\in\NP$
and designated polynomial-length advice $a_n$ such that
$x\in L\iff\langle x,a_n\rangle\in K$. Since $\overline K\in\coNP$,
the hypothesis supplies an $\NP/\poly$ verifier for $\overline K$.
Bundle $a_n$ with the latter verifier's designated advice for every
possible pair-encoding length up to the relevant polynomial bound.
This gives an $\NP/\poly$ verifier for $\overline L$. The total bundle
is polynomial, so $\mathcal D$ is closed under complementation.

The class $\mathcal D$ is also closed under existential polynomial-length
projection: guess the quantified string and an NP witness, and bundle
the designated advice for all relevant combined-input lengths. By
complementation it is closed under universal polynomial-length
projection as well. Starting with polynomial-time predicates and
applying a fixed number of such projections proves the assertion for
all of $\PH$. Complement closure also identifies $\mathcal D$ with
$\coNP/\poly$.
\end{proof}

\begin{corollary}[Deterministic simulation with polynomial advice]
\label[corollary]{cor:ph-advice}
Under $\hyp$,
\begin{equation}
 \PH\subseteq
 \DTIME\bigl(2^{O(\Llog(n)^d)}\bigr)/\poly.
 \label{eq:ph-deterministic-advice}
\end{equation}
The exponent $d$ is independent of the level of $\PH$; the advice
polynomial and constants may depend on the language.
\end{corollary}

\begin{proof}
\Cref{prop:structural} gives $\coNP\subseteq\NP/\poly$. By
\cref{lem:ph-np-advice}, every fixed $L\in\PH$ has a polynomial-time
verifier $V$ and polynomial advice $a_n$ with
\[
 x\in L\iff\exists w\ V(x,a_n,w)=1.
\]
Compile this predicate into a circuit $C_{x,a_n}$ on its polynomially
many witness inputs. Use a fixed layout and hardwire $x,a_n$ in
fixed-width constant slots, without simplification. All circuits at
input length $n$ then have the same length $N(n)=n^{O(1)}$.

Apply the certified-search procedure of \cref{lem:certified-search},
proved independently below, with its designated $\alpha_{N(n)}$.
The deterministic interpreter accepts exactly when
$F(C_{x,a_n},\alpha_{N(n)})\neq\bot$. Negative instances have no
satisfying witness, and positive instances yield one. The advice
$(a_n,\alpha_{N(n)})$ is polynomial, and the time is
$2^{O(\Llog(N(n))^d)}=2^{O(\Llog(n)^d)}$.

Equivalently, one can substitute the same verifier directly into
\cref{prop:advice}; that route gives its more explicit
$O(a(n)+w(n)+n)$ advice accounting. Neither route certifies arbitrary
choices of the original advice $a_n$.
\end{proof}

\section{The completeness obstruction and certified-search elimination}
\label{sec:gap}

\subsection{Why soundness-only refutation advice fails}

The first subsection identifies a failed implication, not a separation
theorem. The source binary theorem explicitly distinguishes correctness
under designated advice from correctness on arbitrary advice--input
pairs~\cite[Section~3 and Appendix~H]{BD26}. The second subsection
avoids that completeness test by changing the advice's role: it assists
in finding a satisfiability witness, which is checked directly. The
resulting two-block predicate is quasipolynomial-time for $d>1$.

From $\coNP\subseteq\NP/\poly$, fix a polynomial-time verifier $V$
and a designated polynomial-size package $a_N$ containing the advice
for every formula length at most $N$. For valid formulas $\varphi$,
\begin{equation}
 \varphi\in\UNSAT\iff
 \exists w\ V(\varphi,a_N,w)=1,
 \qquad |\varphi|\leq N.
 \label{eq:advised-unsat}
\end{equation}
All witnesses and assignment strings below have a fixed polynomial
bound in $N$, with irrelevant trailing positions ignored under canonical
encoding conventions. Quantification over $\varphi$ is over valid
formulas of length at most $N$.

Define
\begin{equation}
 \Sound_N(a)\iff
 \forall\varphi,w,z\quad
 \bigl[V(\varphi,a,w)=1\Longrightarrow\varphi(z)=0\bigr].
 \label{eq:sound}
\end{equation}
This is a coNP predicate in $(1^N,a)$: failure has a formula, an accepted
alleged UNSAT proof, and a satisfying assignment as a polynomial-size
certificate.

Let $L\in\Pi_2^p$, written as
\[
 x\in L\iff\forall u\ \exists v\ R(x,u,v)=1.
\]
Construct $\psi_{x,u}$ such that
$\psi_{x,u}\in\SAT\iff\exists v\ R(x,u,v)=1$.
All these formulas have length at most $N=\poly(|x|)$.
A tempting second-level characterization is
\begin{equation}
 \exists a\left[
   \Sound_N(a)\ \land\
   \forall u,w\ \neg V(\psi_{x,u},a,w)
 \right].
 \label{eq:false-collapse}
\end{equation}
It has the desired $\exists\forall$ shape, but it need not characterize
$L$.

\begin{proposition}[Reject-all advice defeats soundness-only checking]
\label[proposition]{prop:empty}
The designated-advice guarantee \eqref{eq:advised-unsat} does not
justify replacing $x\in L$ by \eqref{eq:false-collapse}.
\end{proposition}

\begin{proof}
Replace the advice format by a one-bit mode header. Under header~$1$,
simulate the original verifier; under header~$0$, reject every proposed
proof. Extend the designated advice with header~$1$, so
\eqref{eq:advised-unsat} continues to hold with only one additional
advice bit. A header-$0$ package $a_{\mathrm{empty}}$ is sound because
it proves nothing. It also satisfies
\[
 \forall u,w\ \neg V(\psi_{x,u},a_{\mathrm{empty}},w)
\]
for every input $x$. Thus \eqref{eq:false-collapse} accepts all inputs.
\end{proof}

The needed completeness assertion is
\begin{equation}
 \Complete_N(a)\iff
 \forall\varphi\in\UNSAT_{\leq N}\ \exists w\
 V(\varphi,a,w)=1.
 \label{eq:complete}
\end{equation}
For clarity, this condition can be written without a restricted domain as
\begin{equation}
 \forall\varphi\ \exists w,z\quad
       [V(\varphi,a,w)=1\ \lor\ \varphi(z)=1].
 \label{eq:complete-pi2}
\end{equation}
For an unsatisfiable formula the second disjunct is impossible, whereas
for a satisfiable one a satisfying assignment makes the assertion true.
This is a $\Pi_2^p$ condition. Existentially guessing $a$ and then
imposing it reintroduces the third quantifier block; the preceding
argument supplies no coNP test for completeness.

\begin{proposition}[A sufficient additional completeness-enforcement lemma]
\label[proposition]{prop:robust}
Suppose there are a polynomial-time $V$, polynomial advice and witness
bounds, and a predicate $G(1^N,a)\in\coNP$ satisfying:
\begin{enumerate}[label=(\roman*)]
 \item for each $N$, some permitted advice $a$ satisfies $G(1^N,a)$;
 \item for every $a$ satisfying $G(1^N,a)$ and every formula
       $|\varphi|\leq N$,
       \[
        \varphi\in\UNSAT\iff\exists w\ V(\varphi,a,w)=1.
       \]
\end{enumerate}
Then $\PH=\Sigma_2^p$.
\end{proposition}

\begin{proof}
For $L\in\Pi_2^p$ and $\psi_{x,u}$ as above,
\[
 x\in L\iff
 \exists a\left[
 G(1^N,a)\land\forall u,w\ \neg V(\psi_{x,u},a,w)
 \right].
\]
On a yes-instance, use an advice package supplied by (i); soundness
ensures that no satisfiable $\psi_{x,u}$ has an accepted UNSAT proof.
On a no-instance, some $\psi_{x,u}$ is unsatisfiable, and completeness
in (ii) gives a proof under every package passing $G$.
Because $G$ is in coNP, its universal witnesses can be combined with
$(u,w)$ after the existential advice block. Hence
$\Pi_2^p\subseteq\Sigma_2^p$, which collapses the hierarchy to that
level.
\end{proof}

\begin{remark}[Status of the sufficient condition]
\cref{prop:robust} is an implication from an \emph{additional} assumption.
No construction of such a predicate $G$ from $\hyp$ is given here.
In particular, common advice for the NP and coNP interpretations of a
binary comparator does not by itself supply this predicate.
\end{remark}

\subsection{Certified search with advice}
\label{subsec:certified-search}

\begin{lemma}[Certified search under arbitrary candidate advice]
\label[lemma]{lem:certified-search}
Under $\hyp$, there are a fixed polynomial-length advice format
$A(N)=O(N^2)$ and a deterministic algorithm $F(C,\alpha)$ with the
following properties. On a circuit $C$ of encoding length $N$ and any
permitted advice string, $F$ runs in time
$2^{O(\Llog(N)^d)}$ and returns an assignment or $\bot$.
\begin{enumerate}[label=(\alph*)]
 \item For every candidate advice $\alpha$, if
       $F(C,\alpha)=v\neq\bot$, then $C(v)=1$.
 \item For every $N$, there is a single $\alpha_N\in\bits^{A(N)}$
       such that $F(C,\alpha_N)\neq\bot$ for every satisfiable
       circuit $C$ of encoding length $N$.
\end{enumerate}
Malformed advice is rejected. The running-time bound holds on all
permitted advice, not merely on the designated sequence.
\end{lemma}

\begin{proof}
Use the witness-producing algorithm $U'$ from \cref{prop:unique}.
Its output is always checked in the circuit passed to it, including
on off-promise inputs.

\paragraph{One advice format for all circuits of length $N$.}
Set
\[
 M=16(N+1),\qquad \rho_N=3N+1,\qquad A(N)=M\rho_N.
\]
Parse a well-formed $\alpha$ as $M$ blocks
$s_1,\ldots,s_M\in\bits^{\rho_N}$. A circuit of length $N$ may
have any number $w\leq N$ of free inputs. When $w\geq1$, use the
first $3w+1$ bits of each $s_i$ as its Toeplitz-hash seed, in the
fixed format of \cref{lem:isolation}. This projection of a uniform
$\rho_N$-bit block is uniform on the required seed space. In
particular, circuits of the same encoding length but different input
arities still share one advice length and parsing convention. No free
dummy witness variables are introduced. When $w=0$, evaluate $C$
directly, returning the empty assignment if true and $\bot$ if false.

For $w\geq1$, for every $i=1,\ldots,M$ and $r=1,\ldots,w+1$,
form $C_{s_i,r}$ as in \cref{lem:isolation} and run $U'$ on it.
Return the first assignment $v$ that it supplies after explicitly
checking both its length and $C_{s_i,r}(v)=1$. If no call supplies
such an assignment, return $\bot$.

\paragraph{Soundness for every candidate advice.}
Every restricted circuit has the same free assignment variables as
$C$ and satisfies
\[
 C_{s_i,r}(v)=1\Longrightarrow C(v)=1.
\]
Thus every returned assignment is genuine, whatever seeds the advice
contains. A malformed string cannot cause acceptance, and an
unsatisfiable input returns $\bot$ for every advice string.

\paragraph{One advice string succeeds simultaneously.}
For a fixed satisfiable circuit $C$, \cref{lem:isolation} shows that
a uniform seed has probability at least $3/16$ of making some prefix
restriction uniquely satisfiable. On that restriction $U'$ returns
the unique witness. Hence $M$ independent uniform blocks miss all
successful restrictions with probability at most
\begin{equation}
 (13/16)^M\leq e^{-3(N+1)}<2^{-N-2}.
 \label{eq:certified-miss}
\end{equation}
There are at most $2^N$ valid circuit encodings of length $N$.
A union bound makes the probability that any satisfiable one is missed
less than $1/4$. Some fixed list therefore works for all of them.
This existence argument requires neither recognizing the good seeds
nor finding the list uniformly.

\paragraph{Time and advice.}
There are $M(w+1)=O((N+1)^2)$ calls to $U'$, each on a circuit of
size $N^{O(1)}$ (with harmless constant adjustments for small $N$).
All construction and verification costs are polynomial. Because
$\Llog(N^{O(1)})=O(\Llog(N))$, the total time is
$2^{O(\Llog(N)^d)}$. The advice length $M(3N+1)$ is $O(N^2)$.
\end{proof}

\paragraph{Optional linear-advice compression.}
The common ambient seed space above also permits an $O(N)$-bit advice
version. Each satisfiable length-$N$ circuit defines a set of density
at least $3/16$ in $\bits^{3N+1}$. By
\cite[Lemma~F.5]{BD26}, one short expander-walk description hits all
at most $2^N$ such sets and has length $3N+1+O(N)=O(N)$.
Reconstructing its polynomially many seeds preserves the clock and
all-advice soundness. The results below use only the elementary
$O(N^2)$ bound of \cref{lem:certified-search}.

\subsection{Eliminating the final existential block}

\begin{theorem}[Polynomial-length two-block simulation]
\label{thm:certified-collapse}
Under $\hyp$,
\begin{equation}
 \PH\subseteq
 \Sigma_2^p\!\left[2^{O(\Llog(n)^d)}\right]
 \cap
 \Pi_2^p\!\left[2^{O(\Llog(n)^d)}\right].
 \label{eq:certified-two-block}
\end{equation}
All quantified strings have length polynomial in the original input
length. The deterministic predicate has exponent $d$, independently
of the language's level in $\PH$. No advice is supplied outside the
quantifier blocks, and no oracle is used by the predicate.
\end{theorem}

\begin{proof}
By \cref{prop:structural}, every $L\in\PH$ belongs to
$\Sigma_3^p$. Choose fixed polynomials and a polynomial-time
predicate $R$ so that, with canonical padding of the three blocks,
\begin{equation}
 x\in L\iff
 \exists y\in\bits^{p(n)}\ \forall z\in\bits^{q(n)}\
 \exists w\in\bits^{r(n)}\ R(x,y,z,w)=1.
 \label{eq:certified-sigma-three}
\end{equation}

For each $n$, uniformly construct a circuit
$\Gamma_n(x,y,z,w)$ computing $R$ on these fixed input lengths.
Hardwire $x,y,z$ in fixed-width constant slots, retaining the entire
layout without simplifying gates. Equal-length encodings of zero
and one constant gates, or fixed-slot syntactic padding, ensure that
all resulting circuits $\varphi_{x,y,z}$ have exactly the same
encoding length $N(n)=n^{O(1)}$. Only $w$ remains free. In particular,
\begin{equation}
 \varphi_{x,y,z}\in\CSAT
 \iff\exists w\ R(x,y,z,w)=1.
 \label{eq:certified-hardwire}
\end{equation}
Padding changes neither this equivalence nor the set of witness
assignments; it introduces no free auxiliary inputs.

With $F$ from \cref{lem:certified-search}, we claim
\begin{equation}
 x\in L\iff
 \exists y\in\bits^{p(n)}\
 \exists\alpha\in\bits^{A(N(n))}\
 \forall z\in\bits^{q(n)}\quad
 F(\varphi_{x,y,z},\alpha)\neq\bot.
 \label{eq:certified-collapse}
\end{equation}
If $x\in L$, choose $y$ from \eqref{eq:certified-sigma-three} and
choose the single $\alpha_{N(n)}$ supplied by
\cref{lem:certified-search}(b). Every circuit in
\eqref{eq:certified-hardwire} is then satisfiable, and that same
advice works for every $z$.

Conversely, suppose the right side of \eqref{eq:certified-collapse}
holds for some $y,\alpha$. For every $z$, part~(a) of the lemma
makes the nonfailure output an actual satisfying assignment of
$\varphi_{x,y,z}$. Equation~\eqref{eq:certified-hardwire} therefore
supplies the required $w$ for every $z$, proving $x\in L$.
No soundness or completeness test for $\alpha$ is needed.

Merge $(y,\alpha)$ into one existential block. Its length is
$p(n)+O(N(n)^2)=\poly(n)$, and the universal block has length
$q(n)=\poly(n)$. Circuit construction takes polynomial time, and
the predicate clock is
\[
 2^{O(\Llog(N(n))^d)}=2^{O(\Llog(n)^d)}.
\]
This proves the existential--universal containment. Apply the same
construction to $\overline L\in\PH$ and negate its deterministic
predicate to obtain a universal--existential characterization of $L$.
\end{proof}

\begin{remark}[What this removes, and what it does not]
The advice in \eqref{eq:certified-collapse} helps produce positive
witnesses; it does not assert the absence of refutations. Advice that
makes search fail cannot turn a false instance into a true one.
For a false instance and each $y$, some $z$ has an unsatisfiable
$\varphi_{x,y,z}$, on which every advice gives $\bot$.
Thus the reject-all construction of \cref{prop:empty} is harmless
here. This avoids the completeness test of \cref{prop:robust}; it
does not construct that test.

For $d=1$ the predicate is polynomial-time, and the theorem gives
$\PH=\Sigma_2^p=\Pi_2^p$, consistent with Sivakumar's
logarithmic-arity consequence $\NP=\class{RP}$~\cite{Siv99}.
For $d>1$, \eqref{eq:certified-two-block} is a two-block simulation
with a longer predicate clock, not a collapse to the usual second
level of $\PH$. Nor is the reverse inclusion of these longer-clock
classes in $\PH$ asserted.
\end{remark}

\subsection{Symmetric alternation with competing seed lists}

\begin{definition}[Symmetric alternation with a longer predicate clock]
\label[definition]{def:s2-clock}
For a time bound $T$, a language $L$ belongs to $\Sclass^p[T]$ if there are
fixed polynomials $p,q$ and a deterministic predicate $P$, clocked in
time $O(T(n))$ on all triples $(x,Y,Z)$ with $|Y|=p(n)$ and
$|Z|=q(n)$, such that
\[
  x\in L\implies\exists Y\,\forall Z\;P(x,Y,Z)=1,\qquad
  x\notin L\implies\exists Z\,\forall Y\;P(x,Y,Z)=0.
\]
As above, exponential $O(\cdot)$ notation denotes a union over its
constant, with time measured in the original input length.
Thus $\Sclass^p[\poly]=\Sclass^p$.
For $d>1$, membership in $\Sclass^p[2^{O(\Llog(n)^d)}]$ is a
longer-clock symmetric simulation, not membership in ordinary
$\Sclass^p$. No reverse containment of the longer-clock class in $\PH$
is asserted.
\end{definition}

\begin{lemma}\label[lemma]{lem:s2-in-sigma2}
$\Sclass^p[T]\subseteq\Sigma_2^p[T]\cap\Pi_2^p[T]$.
\end{lemma}

\begin{proof}
We claim $x\in L\iff\exists Y\,\forall Z\;P(x,Y,Z)=1$. The forward
direction is the definition. If $x\notin L$, the no-certificate $Z^*$
gives $P(x,Y,Z^*)=0$ for every $Y$, so the right side fails. In the
same way, $x\notin L\iff\exists Z\,\forall Y\;P(x,Y,Z)=0$, which is a
universal--existential characterization of $L$ with predicate
$P(x,Y,Z)$.
\end{proof}

Fix a polynomial $Q$. A \emph{seed bundle} for input length $n$ is a
string $\beta=(\beta_1,\dots,\beta_{Q(n)})$ with
$\beta_N\in\bits^{A(N)}$, where $A(N)=16(N+1)(3N+1)$ is the fixed advice length in
\cref{lem:certified-search}. Its length $\sum_{N\le Q(n)}A(N)$ is a fixed polynomial in
$n$. The \emph{designated bundle} has $\beta_N=\alpha_N$ for every
$N\le Q(n)$. For two bundles $\beta,\beta'$ and a circuit $C$ of
encoding length $N\le Q(n)$, define
\[
  \mathcal O_{\beta,\beta'}(C)=1\iff
  F(C,\beta_N)\ne\bot\ \text{ or }\ F(C,\beta'_N)\ne\bot .
\]

\begin{lemma}[Two bundles give an exact oracle]\label[lemma]{lem:two-bundles}
Under $\hyp$, for all bundles $\beta,\beta'$ and every circuit $C$ of
length at most $Q(n)$:
\begin{enumerate}
\item[(a)] if $\mathcal O_{\beta,\beta'}(C)=1$, then $C$ is satisfiable;
\item[(b)] if $\beta$ or $\beta'$ is the designated bundle, then
$\mathcal O_{\beta,\beta'}(C)=1$ exactly when $C$ is satisfiable.
\end{enumerate}
Each evaluation takes time $2^{O(\Llog(n)^d)}$ on all bundles.
Malformed circuit encodings, including the empty encoding, receive answer
zero before any bundle component is indexed.
\end{lemma}

\begin{proof}
Part (a) is \cref{lem:certified-search}(a) applied to both calls. For (b), a satisfiable
$C$ receives a witness from the designated component by \cref{lem:certified-search}(b),
and an unsatisfiable $C$ receives $\bot$ from both calls by (a). The
clock of \cref{lem:certified-search} holds on all permitted advice, and
$\Llog(Q(n))=O(\Llog(n))$.
\end{proof}

\begin{theorem}[Symmetric simulation of the polynomial hierarchy]
\label{thm:s2-simulation}
Under $\hyp$,
\[
  \PH=\Sclass^{\NP}\subseteq
  \Sclass^p\bigl[2^{O(\Llog(n)^d)}\bigr].
\]
All certificates have length polynomial in the original input length,
the predicate is deterministic and uses no oracle, and the exponent $d$
does not depend on the level of $\PH$.
\end{theorem}

\begin{proof}
By \cref{prop:structural}, $\PH=\Sclass^{\NP}$. Fix
$L\in \Sclass^{\NP}$ with a clocked polynomial-time oracle verifier
$V$, whose clock holds for every oracle-answer sequence, and certificate lengths $p_0(n),q_0(n)$. Replace its fixed $\NP$ oracle by a polynomial-time many-one
reduction to $\CSAT$; in particular, a SAT oracle is handled by the
usual formula-to-circuit map.
Because $V$ is clocked, a polynomial $Q$ bounds the encoding length of
every query made on any triple $(x,y,z)$ with $|x|=n$.

The new certificates are $Y=(y,\beta)$ and $Z=(z,\beta')$, where $y$
and $z$ are certificates of $V$ and $\beta,\beta'$ are seed bundles for
length $n$. The predicate $P(x,Y,Z)$ parses both certificates, treating
a malformed bundle component as advice on which $F$ returns $\bot$, and
simulates $V(x,y,z)$, answering each query $C$ by
$\mathcal O_{\beta,\beta'}(C)$. It is clocked by the polynomial clock
of $V$ and polynomially many evaluations as in
Lemma~\ref{lem:two-bundles}, for a total of $2^{O(\Llog(n)^d)}$ on all
triples.

Let $x\in L$, let $y^*$ win against every opposing certificate of $V$,
and put $Y^*=(y^*,\beta^*)$ with $\beta^*$ designated. Fix any
$Z=(z,\beta')$. By induction on the number of queries, the simulated
run and the exact-oracle run of $V^{\CSAT}(x,y^*,z)$ make
the same queries: as long as the histories agree, the next query is the
same, and Lemma~\ref{lem:two-bundles}(b) makes its simulated answer
exact. Hence $P(x,Y^*,Z)=V^{\CSAT}(x,y^*,z)=1$. The case
$x\notin L$ is symmetric, with the no-prover supplying its winning $z^*$
and the designated bundle. The query induction is necessary because the
queries can depend on both original certificates and on earlier answers;
completeness of the designated bundle holds for all circuits within the
length bound, not merely for a preselected list of queries.
\end{proof}

\begin{corollary}
\Cref{thm:certified-collapse} follows from Theorem~\ref{thm:s2-simulation} and
Lemma~\ref{lem:s2-in-sigma2}. This route uses neither the containment
$\Sclass^{\NP}\subseteq\Sigma_3^p\cap\Pi_3^p$ nor the third-level
form~\eqref{eq:certified-sigma-three}.
\end{corollary}

\begin{remark}[What each bundle protects]
A dishonest bundle cannot produce a false positive answer, because every
positive answer carries a checked witness. It cannot suppress a true
positive answer either, because the honest bundle supplies the witness.
Each prover therefore needs only its own designated advice, and neither
bundle is tested for completeness. A refutation-advice interpreter
alone does not provide the same construction: it supplies existential
certificates, not a procedure that extracts them with all-advice soundness
within the desired clock. Certificates for a predetermined query list do
not cover arbitrary queries depending on the opponent's certificate.
This diagnoses that attempted simulation, not an impossibility theorem
for every use of refutation advice.
\end{remark}

\begin{remark}[Polynomial-time replacement]\label[remark]{rem:s2-poly}
Replace $F$ by a polynomial-time certified search $F_0$ as in
\cref{prop:certified-poly}(ii), replacing $A(N)$ in each bundle by
the corresponding polynomial advice bound of $F_0$. Summing that bound
over $N\le Q(n)$ still gives polynomial-length certificates. The same
proof then gives
$\Sclass^{\NP}\subseteq \Sclass^p$. Since $\SAT\in\Pclass/\poly$
implies $\NP\subseteq\coNP/\poly$, \cref{thm:yap} then gives
$\PH=\Sclass^p\subseteq\ZPP^{\NP}$~\cite{cai-s2zpp},
the known symmetric-alternation form of the Karp--Lipton collapse. The
seed components of the certificates in Theorem~\ref{thm:s2-simulation}
depend only on $n$, but the components $y,z$ do not; no oblivious
version is claimed.
\end{remark}

\section{Weak prediction and the information in exclusions}
\label{sec:learning}

\subsection{Distribution-universal weak prediction already gives circuits}

\begin{proposition}[Weak prediction versus exact circuits]
\label[proposition]{prop:minimax}
For a language $A$, the following are equivalent.
\begin{enumerate}[label=(\roman*)]
 \item $A\in\Pclass/\poly$.
 \item There are a fixed polynomial circuit-size bound $s(n)$ and a
       function $0<\gamma(n)\leq1/2$ with
       $1/\gamma(n)\leq\poly(n)$ such that, for every $n$ and every
       distribution $\mu$ on $\bits^n$, some circuit $C_\mu$ of size
       at most $s(n)$ satisfies
       \[
        \Pr_{x\sim\mu}[C_\mu(x)=\chi_A(x)]
            \geq\tfrac12+\gamma(n).
       \]
\end{enumerate}
No efficient method for finding $C_\mu$ from $\mu$ is assumed.
\end{proposition}

\begin{proof}
An exact polynomial-size circuit proves (i)$\Rightarrow$(ii).
For the converse, fix $n$ and consider the finite zero-sum game between
size-$s(n)$ circuits and inputs, with payoff one for a correct answer
and zero for an incorrect one. Finite minimax, in the program--input
form used in~\cite[Section~5.3]{BD13}, supplies a distribution $\nu$
over these circuits such that
\[
 \Pr_{C\sim\nu}[C(x)=\chi_A(x)]\geq\tfrac12+\gamma(n)
 \quad\text{for every }x\in\bits^n.
\]
Choose an odd integer $M=O((n+1)/\gamma(n)^2)$ and sample $M$
circuits independently from $\nu$. For a fixed input, a Chernoff bound
makes the probability that their majority is wrong at most
$\exp(-2M\gamma(n)^2)$. Choose the constant in $M$ to make this less
than $2^{-n-1}$. A union bound over $2^n$ inputs shows that some fixed
sampled majority is correct everywhere. Hardwire these circuits and
the majority gate into one circuit of polynomial size. This proves
$A\in\Pclass/\poly$.
\end{proof}

The hypothesis of \cref{prop:minimax}(ii), specialized to SAT, would
therefore give the small-circuit conclusion itself. It is not a weaker
proved intermediate consequence of superlogarithmic comparability.
By contrast, the distributional theorem in
\cite[Theorem~5.2.1]{BD13} gives success
\begin{equation}
 \frac{2^{k-1}}{2^k-1}
   =\frac12+\frac{1}{2(2^k-1)},
 \label{eq:thesis-advantage}
\end{equation}
with short comparator-based advice. Its advantage becomes smaller than
every inverse polynomial when $k(n)=\omega(\log n)$.

\subsection{A local exclusion-only experiment}

The next experiment concerns only the information in an excluded
vector. The relativized construction in \cref{sec:oracle} separately
addresses oracle-polynomial-time comparators.

\begin{proposition}[Posterior bias from one or several exclusions]
\label[proposition]{prop:posterior}
Let $X$ be uniform in $\bits^k$. Conditional on $X$, sample an excluded
vector $B$ uniformly from $\bits^k\setminus\{X\}$.
After observing $B=b$, the best prediction of any specified coordinate
$X_i$ has success probability
\[
 \frac{2^{k-1}}{2^k-1}.
\]
More generally, after $q<2^k$ conditionally independent such exclusions,
the posterior advantage over $1/2$ for any specified coordinate is at
most
\begin{equation}
 \frac{q}{2(2^k-q)}.
 \label{eq:posterior-q}
\end{equation}
\end{proposition}

\begin{proof}
For a fixed observed $b$, Bayes' rule makes $X$ uniform on the
$2^k-1$ remaining vectors. Among them, $2^{k-1}$ have
$X_i=1-b_i$, and $2^{k-1}-1$ have $X_i=b_i$. Predicting $1-b_i$
is optimal and gives the stated success.

For the general case, let $E$ be the set of distinct observed excluded
vectors and put $e=|E|\leq q$. Conditional independence gives equal
likelihood $(2^k-1)^{-q}$ for every possible $X\notin E$ and likelihood
zero for $X\in E$. Thus the posterior is uniform outside $E$.
Let $e_0,e_1$ count the vectors of $E$ with coordinate $i$ equal to
zero and one, respectively. The difference between the two remaining
coordinate counts is $|e_0-e_1|\leq e$. Hence the posterior advantage
is
\[
 \frac{|e_0-e_1|}{2(2^k-e)}
 \leq\frac{e}{2(2^k-e)}
 \leq\frac{q}{2(2^k-q)}.
\]
\end{proof}

For polynomial $q(n)$ and $k(n)=\omega(\log n)$,
\eqref{eq:posterior-q} is superpolynomially small. This shows why a
local interpretation of arbitrary exclusions need not produce the
inverse-polynomial advantage required by \cref{prop:minimax}.
A stronger SAT-specific argument could still exploit the syntax,
self-reducibility, reductions, or actual program of the comparator;
none of those are present in this experiment.

\subsection{The exact polynomial-time replacement condition}

\begin{proposition}[Polynomial-time certified search versus small circuits]
\label[proposition]{prop:certified-poly}
The following statements are equivalent, without assuming $\hyp$.
\begin{enumerate}[label=(\roman*)]
 \item $\SAT\in\Pclass/\poly$.
 \item There are a fixed polynomial $q$ and a deterministic
       polynomial-time procedure $F_0(C,\beta)$ that returns an
       assignment or $\bot$, is sound under every permitted advice
       $\beta\in\bits^{q(N)}$, and has one designated $\beta_N$
       succeeding on every satisfiable length-$N$ circuit.
 \item Circuit satisfiability has polynomial-size search circuits:
       a circuit family of polynomial size outputs a satisfying
       assignment on every satisfiable input circuit.
\end{enumerate}
In (iii), arbitrary outputs on unsatisfiable inputs may be filtered
by direct witness verification.
\end{proposition}

\begin{proof}
For (ii)$\Rightarrow$(i), hardwire $\beta_N$ into the polynomial-time
predicate $[F_0(C,\beta_N)\neq\bot]$. Soundness and completeness
make it an exact circuit-satisfiability decider with polynomial advice,
hence with polynomial-size circuits. Restricting to formula inputs
or applying the standard formula-to-circuit map proves (i).
Hardwiring and compiling the witness-producing computation also
proves (ii)$\Rightarrow$(iii).

For (i)$\Rightarrow$(ii), polynomial-time reductions imply
$\CSAT\in\Pclass/\poly$. Fix a polynomial $b$ bounding the
encoding lengths of all restrictions used when self-reducing an
input circuit of length $N$. Supply as advice the designated
circuit-satisfiability decision circuits for all lengths at most
$b(N)$. There are polynomially many such circuits, each of
polynomial size, so the total advice length is some fixed polynomial
$q(N)$.

The interpreter searches one assignment bit at a time. It tests the
zero restriction with the appropriate advised circuit and chooses
zero if that circuit answers yes; otherwise it chooses one. With
the designated decision circuits, a satisfying extension is
preserved at every step on every satisfiable input. At the end it
checks the proposed assignment directly in the original circuit,
returning it only if it satisfies the circuit. For malformed advice
it returns $\bot$; well-formed but incorrect advice can cause failure
but cannot defeat the final check. Circuits supplied as advice are
explicit, acyclic, size-bounded descriptions, so the running time
is polynomial for all permitted advice.

Finally, for (iii)$\Rightarrow$(ii), supply the appropriate search
circuit as advice, evaluate it with a polynomial-time universal
circuit evaluator, and check its output directly. Correct search
circuits give completeness, and the output check gives all-advice
soundness. Fixed padding handles different numbers of assignment
bits and the ground-circuit case.
\end{proof}

This is an equivalence for the \emph{existence} of a polynomial-time
certified-search replacement with \emph{some polynomial} advice bound.
It neither accelerates the particular list-reconstruction program of
\cref{prop:unique} nor promises the specific quadratic or linear
advice bounds of \cref{lem:certified-search}. The equivalence concerns this certified-search route to
$\PH=\Sigma_2^p$.

With such a replacement, the proof of
\eqref{eq:certified-collapse} applies to every $L\in\Sigma_3^p$
with a polynomial-time predicate, without needing $\hyp$ first.
It gives $\Sigma_3^p\subseteq\Sigma_2^p$, and therefore the usual
second-level hierarchy collapse; Remark~\ref{rem:s2-poly} strengthens
this to $\PH=\Sclass^p\subseteq\ZPP^{\NP}$. In this certified-search route,
\cref{prop:certified-poly} and \cref{prop:minimax} identify the
same additional small-circuit condition in different forms.

\subsection{What the resource bounds do not establish}

\paragraph{Lists and evaluation points.}
The choice $m\geq32K$ in \cref{prop:unique} is a convenient sufficient
parameter choice, not a necessity theorem of Sauer--Shelah. With
$m=\Theta(K)$, its general binomial shortlist bound permits
$2^{\Theta(K)}$ surviving values. Such a bound is tight for arbitrary
set systems, but no matching lower bound for the particular
SAT-derived lists is proved here.
If $t$ evaluation points each have at most $\ell$ candidate values,
then the worst-case bound $L\leq t\ell$ makes
$t>2\ell D$ a sufficient inequality for the stated Sudan threshold.
Merely selecting fewer evaluation points supplies no guaranteed
improvement in $\ell$. This diagnoses the cost of the present explicit
shortlisting analysis; it does not rule out other advice,
reconstruction methods, or circuit-specific structure.

\paragraph{Why the uniform randomized exponent is unchanged.}
\Cref{thm:s2-simulation} bounds a deterministic predicate, not a uniform
bounded-error algorithm for the whole symmetric game. The proved uniform
randomized bound remains \cref{thm:ph}, with exponent $d^2$.
The square arises at one identifiable explicit-description step.

For the application of \cref{lem:second} under $\hyp$, choose the
particular SAT algorithm of \cref{cor:random-sat}, not an arbitrary
randomized algorithm satisfying the lemma's time hypothesis. Its only
coins are polynomial-length isolation seeds, followed by deterministic
reconstruction. The amplified tape $r$ of $\widehat B$ in that application
therefore has polynomial length: there are $O(p(n)+1)$ repetitions, each
using polynomially many seed bits. Thus $D_{x,r}$ has a polynomial-length
\emph{succinct program description} $(x,r)$ with the fixed algorithm and
clock implicit, while its explicit circuit description can have size
$2^{O(\Llog(n)^d)}$.

Applying the reconstruction of \cref{prop:unique} to this succinctly
described circuit would require comparing instances of the language
\begin{equation}
 \begin{split}
 \mathcal A=\bigl\{(x,r,h,t,j):\exists u\in\bits^{p(n)}\;
  [&\widehat B(x,u;r)=1\\[-2pt]
   &\land\operatorname{bit}_j(P_u(t)\bmod h)=1]\bigr\},
 \end{split}
 \label{eq:succinct-auxiliary}
\end{equation}
with the same monic-modulus conventions as in \eqref{eq:bit-language}.
The instance fields and the witnesses have polynomial length, but the
verifier uses quasipolynomial time. The usual polynomial-time
Cook--Levin reduction for an $\NP$ verifier is therefore not available
on these succinct inputs. Unrolling this longer computation gives an
\emph{explicit} SAT instance of length $2^{O(\Llog(n)^d)}$; at that
length $\hyp$ guarantees comparison arity only
$O(\Llog(n)^{d^2})$. The field degree in \eqref{eq:field-size} grows
accordingly.

On a tape $r$ satisfying the good-tape condition in
\eqref{eq:common-tape} for a fixed $x$, membership in
\eqref{eq:succinct-auxiliary} agrees with the $\Sigma_2^p$ predicate
\[
 \exists u\ \forall v\quad
 [R(x,u,v)=1\land
  \operatorname{bit}_j(P_u(t)\bmod h)=1].
\]
This is agreement on \emph{good-tape instances}, not a proof that the
full language $\mathcal A$, including bad tapes, belongs to $\Sigma_2^p$.
A hypothesis supplying arity $O(\Llog(n)^d)$ in the succinct input length
for this auxiliary language and its isolation-restricted variants,
or with a carefully specified matching promise, would remove this
explicit-length loss and allow the same reconstruction at exponent $d$.
Such a hypothesis is different from $\hyp$ and has not been derived here.

Neither random-tape length nor sampling the isolation seed list causes
the square. Equation~\eqref{eq:certified-miss} shows that a uniformly
sampled seed list works simultaneously for every satisfiable length-$N$
circuit with probability greater than $3/4$. The separate original
$\NP/\poly$ advice $a_n$ in \cref{cor:ph-advice} is not an isolation seed
list and has no efficient sampler supplied by this proof. The
oracle lower bound of \cref{sec:oracle} concerns the randomized SAT
exponent $d$, rather than the level-independent PH exponent $d^2$.

\paragraph{Why a coNP advice predicate is not interchangeable.}
Suppose an advised coNP decision rule represents satisfiability as
$D_a(C)=1\iff\forall u\ V(C,a,u)=1$. Under arbitrary candidate
advice, the condition that this rule never accepts an unsatisfiable
circuit is
\[
 \forall C\ \exists u,v\quad
       [V(C,a,u)=0\ \lor\ C(v)=1],
\]
with all strings polynomially bounded. This is a $\Pi_2^p$ condition,
not a universal polynomial-time verification supplied by the advice
theorem. Direct checking of the positive witnesses returned by $F$
avoids that condition by using positive certificates that can be
checked directly.

\section{A relativized limit on the reconstruction exponent}
\label{sec:oracle}

The reconstruction and isolation arguments in
\cref{prop:unique,cor:random-sat} relativize when the reductions compile
explicit oracle computations and every satisfiability instance is
interpreted relative to the same oracle. The polynomial length overhead
preserves the power of $\Llog(n)$.
For each fixed admissible arity function $K$, the construction below
provides an oracle in which the complete set is $K$-membership
comparable and $\NP=\coNP$, but a unary language in NP requires
bounded-error time on the scale of $2^{K(n)}$. For
$K(n)=\Theta(\Llog(n)^d)$, this gives a matching logarithmic-power
lower bound for the randomized SAT upper bound.

\subsection{Conventions}

For an oracle $A$, let
\[
  \mathrm{K}^A=\bigl\{\langle M,x,1^t\rangle:
  M^A \text{ has an accepting path on } x
  \text{ of length at most } t\bigr\},
\]
where $M$ ranges over nondeterministic oracle machines.
Fix a prefix-free code for $M$ and encode the triple as
$\operatorname{code}(M)1^t0x$. This is polynomial-time parseable and
has length $|\operatorname{code}(M)|+t+1+|x|\ge t$, so a computation
certifying membership of an instance of length at most $m$ queries only
strings of length at most $m$. Every language in $\mathrm{NP}^A$ reduces
to $\mathrm{K}^A$ by the oracle-independent map
$x\mapsto\langle M,x,1^{p(|x|)}\rangle$. Relativized membership comparability is the definition in
\cref{sec:intro}, applied to the language $\mathrm K^A$, with a
deterministic polynomial-time $A$-oracle comparator. Oracle query tapes
are initially blank, and writing a query is charged to the running time.
Malformed instance encodings are outside $\mathrm K^A$.

An \emph{arity function} is a function $K:\mathbb N\to\mathbb N$ that is
computable in time polynomial in $n$, nondecreasing, satisfies
$K(n)\ge 2$, and satisfies $K(n)\le n/2$ for all sufficiently large $n$.
Every such $K$ is polynomially bounded (including its finitely many
exceptional values). Write $\hyp^A$ for the complete-set formulation
with threshold $\lceil c_0\Llog(n)^d\rceil$. For this application take
$K(n)=\max\{2,\lceil c_0\Llog(n)^d\rceil\}$ and handle the finitely
many smaller thresholds directly. Constants in an exact threshold refer
to this fixed complete-set encoding; polynomial reductions between
encodings preserve the logarithmic power $d$.

\emph{Oracle layout.} The oracle has the form
\[
  A=\{0u:u\in R\}\ \cup\
  \bigl\{1\langle 1^m,\tau,i\rangle:\ Q_m(\tau)_i=1\bigr\},
\]
where $R\subseteq\{0,1\}^*$, and for each $m$ the \emph{level-$m$ table}
$Q_m$ assigns a vector $Q_m(\tau)\in\{0,1\}^{K(m)}$ to every
$K(m)$-tuple $\tau$ of instance encodings of length at most $m$. The
pairing is polynomial-time parseable, has length polynomial in
$m+|\tau|+\log(i+1)$, and satisfies
$|\langle 1^m,\tau,i\rangle|\ge m+1$. Thus each level-$m$ table
string has length at least $m+2$. The fields must obey the indicated
arity and instance-length bounds and $1\le i\le K(m)$; all other
strings have oracle value zero.

\begin{lemma}[Levels]\label[lemma]{lem:levels}
Whether an instance of length at most $m$ belongs to $\mathrm{K}^A$ is
determined by $R\cap\{0,1\}^{\le m}$ and the tables $Q_{m'}$ with
$m'<m$.
\end{lemma}

\begin{proof}
All queries have length at most $m$. A table string of length at most
$m$ encodes a level $m'$ with $m'+2\le m$.
\end{proof}

For a level-$m$ tuple $\tau=(y_1,\dots,y_{K(m)})$ write
\[
  \chi_m(\tau)=\bigl(\chi_{\mathrm K^A}(y_1),\dots,
  \chi_{\mathrm K^A}(y_{K(m)})\bigr).
\]
By Lemma~\ref{lem:levels} this vector is determined by
$R\cap\{0,1\}^{\le m}$ and the tables below level $m$.

\subsection{The random oracle}

\begin{definition}[Distribution]\label[definition]{def:oracle-dist}
Choose the following independently.
\begin{itemize}
\item For each $n$: a uniform bit $\beta_n$ and a uniform
$s_n\in\{0,1\}^n$. Put $R\cap\{0,1\}^n=\{s_n\}$ if $\beta_n=1$ and
$R\cap\{0,1\}^n=\emptyset$ otherwise.
\item For each level $m$ and each level-$m$ tuple $\tau$: a uniform
vector $U_m(\tau)\in\{0,1\}^{K(m)}$ and a uniform nonzero vector
$W_m(\tau)\in\{0,1\}^{K(m)}$.
\end{itemize}
Define the tables by induction on $m$:
\[
  Q_m(\tau)=
  \begin{cases}
    U_m(\tau) & \text{if } U_m(\tau)\ne\chi_m(\tau),\\
    \chi_m(\tau)\oplus W_m(\tau) & \text{otherwise.}
  \end{cases}
\]
By Lemma~\ref{lem:levels}, $\chi_m(\tau)$ is determined by $R$ and the
tables already defined, so the induction is well founded.
\end{definition}

For each $m$ let
\begin{align*}
  \mathcal G_m&=\sigma\bigl(R\cap\{0,1\}^{\le m},\ U_{m'},W_{m'}:m'<m\bigr),\\
  \mathcal F_m&=\sigma\bigl(R\cap\{0,1\}^{<m},\ U_{m'},W_{m'}:m'<m\bigr).
\end{align*}
The tables below level $m$ are measurable with respect to
$\mathcal F_m$, and every vector $\chi_m(\tau)$ is measurable with
respect to $\mathcal G_m$.

\begin{lemma}[Conditional law of the table]\label[lemma]{lem:table-law}
Conditioned on $\mathcal G_m$, the entries $Q_m(\tau)$ are independent,
and each is uniform on $\{0,1\}^{K(m)}\setminus\{\chi_m(\tau)\}$.
\end{lemma}

\begin{proof}
The pairs $(U_m(\tau),W_m(\tau))$ are independent of $\mathcal G_m$ and
of one another. Given $\mathcal G_m$, the value $\chi=\chi_m(\tau)$ is
fixed. On the event $U_m(\tau)\ne\chi$, the entry is uniform on the
complement of $\chi$. On the complementary event it equals
$\chi\oplus W_m(\tau)$, which is also uniform on that complement.
\end{proof}

\subsection{The theorem}

\begin{theorem}[A relativized limit]\label{thm:oracle}
Let $K$ be an arity function. With probability~$1$ over the
distribution of Definition~\ref{def:oracle-dist}, the oracle $A$
satisfies:
\begin{enumerate}
\item[(a)] $\mathrm{K}^A$ is $K(n)$-membership comparable by a
deterministic polynomial-time oracle comparator;
\item[(b)] $\mathrm{NP}^A=\mathrm{coNP}^A$;
\item[(c)] for every time-constructible $T$ with $T(n)\ge n$ and
$T(n)\,2^{-K(n)}\to0$,
\[
  L_R=\bigl\{1^n: R\cap\{0,1\}^n\ne\emptyset\bigr\}
  \in\mathrm{NP}^A\setminus\mathrm{BPTIME}^A\bigl(T(n)\bigr).
\]
\end{enumerate}
Part~(a) holds for every oracle in the support of the distribution.
\end{theorem}

\begin{proof}[Proof of~(a)]
On a tuple $(y_1,\dots,y_t)$ with $n=\max_i|y_i|$ and $t\ge K(n)$, the
comparator forms the level-$n$ tuple $\tau=(y_1,\dots,y_{K(n)})$,
queries the $K(n)$ bits of $Q_n(\tau)$, and outputs them followed by
$t-K(n)$ zeros. The first $K(n)$ coordinates already differ from the
true membership vector, since every table entry avoids
$\chi_n(\tau)$. Repeated instances are allowed, and the running time is
polynomial in the tuple length because $K$ is polynomially bounded
and the query encoding has polynomial length.
\end{proof}

\begin{proof}[Proof of~(b)]
Fix a nondeterministic oracle machine $M_{\mathrm{acc}}$ that accepts
within one step without making a query, and put $z_w=\langle M_{\mathrm{acc}},w,1^1\rangle$. Then
$z_w\in\mathrm{K}^A$ for every oracle, and $|z_w|=|w|+c_1$ for a
constant $c_1$.

\emph{Certificates.} For sufficiently large $m$, let $x$ be an instance of length $m$ and put
$k=K(m)$. A certificate for $x\notin\mathrm K^A$ consists of $k-1$
distinct strings $w_1,\dots,w_{k-1}$ of length $m-c_1$ such that
$Q_m(z_{w_1},\dots,z_{w_{k-1}},x)=1^{k}$. The verifier checks the
format and makes $k$ oracle queries.

\emph{Soundness, for every oracle in the support.} The companions
belong to $\mathrm K^A$. If $x\in\mathrm K^A$, the true vector would be
$1^k$, which the table entry avoids. Hence the certificate forces
$x\notin\mathrm K^A$.

\emph{Completeness, with probability~$1$.} Partition the $2^{m-c_1}$
companion strings into
$N_m=\lfloor 2^{m-c_1}/(k-1)\rfloor$ disjoint groups of size $k-1$. Fix
$x$ of length $m$. Its membership and all the resulting truth vectors
are $\mathcal G_m$-measurable. By Lemma~\ref{lem:table-law}, conditioned
on $\mathcal G_m$ with $x\notin\mathrm K^A$, the $N_m$ entries are
independent, and each equals $1^k$ with probability $1/(2^k-1)$. Thus
\[
  \Pr\bigl[x\notin\mathrm K^A \text{ and } x \text{ has no certificate}\bigr]
  \le \Bigl(1-\frac{1}{2^k}\Bigr)^{N_m}
  \le \exp\bigl(-N_m2^{-k}\bigr).
\]
For large $m$, $k\le m/2$ and
$2^{m-c_1}/(k-1)\ge2$, so
\[
 N_m2^{-k}\ge
 \frac{2^{m-c_1-k}}{2(k-1)}
 \ge\frac{2^{m/2-c_1}}{m}.
\] A union bound over the at most $2^{m+1}$
instances of length $m$ gives a summable sequence in $m$. By the
Borel--Cantelli lemma, with probability~$1$ only finitely many
unsatisfiable instances lack certificates; hardwire their finite set into the verifier, together with a truth
table for the finitely many short lengths. This is a fixed machine for
each selected oracle, not an advice string varying with the input length.

Hence $\overline{\mathrm K^A}\in\mathrm{NP}^A$. Since every
$\mathrm{coNP}^A$ language reduces to $\overline{\mathrm K^A}$ by an
oracle-independent map, $\mathrm{coNP}^A\subseteq\mathrm{NP}^A$, and
therefore $\mathrm{NP}^A=\mathrm{coNP}^A$.
\end{proof}

For part~(c), clearly $L_R\in\mathrm{NP}^A$: guess $s$ and query $0s$.
The lower bound rests on one estimate per input length.

\begin{lemma}[One input length]\label[lemma]{lem:one-length}
Let $M$ be a probabilistic oracle machine that halts within $t(n)\ge n$
steps on every input, oracle and coin sequence. Let $C_n$ be the event
that $M^A(1^n)$ outputs $\chi_{L_R}(1^n)$ with probability at least
$2/3$ over its coins. Then, for every $n$ and almost surely,
\[
  \Pr\bigl[C_n\mid\mathcal F_n\bigr]
  \le \frac12+6\,t(n)\bigl(2^{-K(n)}+2^{-n}\bigr).
\]
\end{lemma}

\begin{proof}
Condition on $\mathcal F_n$ throughout. The remaining randomness
consists of $\beta_n$, $s_n$, $R$ above length $n$, all $U_\ell,W_\ell$
with $\ell\ge n$, and the coins of $M$. All of it is independent of
$\mathcal F_n$.

\emph{Hybrid oracles.} Let $\tilde A$ agree with $A$ on the
$R$-part and on all tables below level $n$, and let every table entry at
a level $\ell\ge n$ be $U_\ell(\tau)$. Let $\tilde A_0$ and $\tilde A_1$
be obtained from $\tilde A$ by setting $R\cap\{0,1\}^n$ to $\emptyset$
and to $\{s_n\}$, respectively, so that $\tilde A=\tilde A_{\beta_n}$.
Put
\[
  a=\Pr_{\mathrm{coins}}\bigl[M^A(1^n)=1\bigr],\qquad
  \tilde a_b=\Pr_{\mathrm{coins}}\bigl[M^{\tilde A_b}(1^n)=1\bigr].
\]

\emph{Step 1: replacing the tables.} The oracles $A$ and $\tilde A$
differ only on the strings of entries $\sigma$ at levels $\ell\ge n$
with $U_\ell(\sigma)=\chi_\ell(\sigma)$; call these entries
\emph{collisions}. With the same coins, the runs of $M$ on $A$ and on
$\tilde A$ coincide until the run on $\tilde A$ first queries a string of
a collision.

Fix an entry $\sigma$ at level $\ell\ge n$ and an index $j$. Let
$\tau_j$ be the $j$-th distinct entry at levels $\ge n$ touched by the
run on $\tilde A$, with a null value if fewer than $j$ entries are touched.
Condition on every primitive random variable other than
$U_\ell(\sigma)$, including the coins. The event $\{\tau_j=\sigma\}$
is then fixed: the raw-oracle computation can be simulated until its
first query to $\sigma$ without using that entry. The vector
$\chi_\ell(\sigma)$ is also fixed, because it depends only on $R$ and
tables at levels strictly below $\ell$. This remains true when the
computation has previously touched higher-level entries: in the hybrid
those answers are their raw, independent $U$ values, not the corrected
tables. Therefore, conditionally on $\mathcal F_n$,
\begin{align*}
 &\Pr[\tau_j=\sigma,\ U_\ell(\sigma)=\chi_\ell(\sigma)
       \mid\mathcal F_n]\\
 &\hspace{12mm}=
 2^{-K(\ell)}\Pr[\tau_j=\sigma\mid\mathcal F_n]\\
 &\hspace{12mm}\le
 2^{-K(n)}\Pr[\tau_j=\sigma\mid\mathcal F_n],
\end{align*}
because $K$ is nondecreasing. Summing over $\sigma$ and over the at most
$t(n)$ indices $j$ shows that the runs differ with probability at most
$t(n)2^{-K(n)}$, jointly over the coins and the remaining oracle
randomness. Let $c$ be the probability over the coins that the runs on
$A$ and on $\tilde A$ differ. Then $|a-\tilde a_{\beta_n}|\le c$ and
$\mathbb E[c\mid\mathcal F_n]\le t(n)2^{-K(n)}$.

\emph{Step 2: hiding the planted string.} The oracles $\tilde A_0$ and
$\tilde A_1$ differ only at the string $0s_n$. Let
$h=\Pr_{\mathrm{coins}}[M^{\tilde A_0}(1^n)\text{ queries }0s_n]$. Then
$|\tilde a_1-\tilde a_0|\le h$. The oracle $\tilde A_0$ is determined
by $\mathcal F_n$, by $R$ above length $n$, and by the tables $U_\ell$
with $\ell\ge n$, so it is independent of $s_n$. A run makes at most
$t(n)$ queries, hence $\mathbb E[h\mid\mathcal F_n]\le t(n)2^{-n}$.

\emph{Step 3: the bound.} Put $\Delta=c+h$, so that
$|a-\tilde a_0|\le\Delta$. On $C_n\cap\{\Delta<1/6\}$, if $\beta_n=1$
then $a\ge2/3$ and $\tilde a_0>1/2$, while if $\beta_n=0$ then
$a\le1/3$ and $\tilde a_0<1/2$. The quantity $\tilde a_0$ is a function
of variables independent of $\beta_n$, and $\beta_n$ is uniform. Hence
this event has conditional probability at most $1/2$. By Markov's
inequality,
$\Pr[\Delta\ge1/6\mid\mathcal F_n]\le6\,\mathbb E[\Delta\mid\mathcal F_n]$,
which gives the stated bound.
\end{proof}

\begin{proof}[Proof of~(c)]
It suffices to treat one probabilistic oracle machine $M$ with one
clock $t=aT$ for a constant $a$, since there are countably many pairs $(M,a)$ and countably many
time-constructible functions. Time constructibility permits an explicit
clock that applies to all oracles and tapes with constant-factor
overhead, so every machine counted by $\mathrm{BPTIME}^A(T)$ is
covered by such a clocked machine. Then
$t(n)2^{-K(n)}\to0$, and $t(n)2^{-n}\to0$ because $K(n)\le n/2$
eventually. Put
$\varepsilon_n=6\,t(n)(2^{-K(n)}+2^{-n})$.

Choose lengths $n_1<n_2<\cdots$ with $\varepsilon_{n_i}\le1/4$ and
$n_{i+1}>t(n_i)$. On input $1^{n_j}$ the machine queries only strings of
length at most $t(n_j)<n_{i}$ for $j<i$. These are $R$-strings below
length $n_i$ and table strings of levels at most $n_i-2$. Hence the
events $C_{n_1},\dots,C_{n_{i-1}}$ are $\mathcal F_{n_i}$-measurable,
and Lemma~\ref{lem:one-length} gives
\[
  \Pr\Bigl[\,\bigcap_{j\le i}C_{n_j}\Bigr]
  \le\Bigl(\frac34\Bigr)^{i}\longrightarrow0.
\]
If $M^A$ decided $L_R$ with bounded error, every $C_{n_j}$ would hold.
So this happens with probability~$0$. A countable union of null events
is null, and together with~(a) and~(b) this proves the theorem. The product argument follows the random-oracle diagonalization
method of Bennett and Gill~\cite{bennett-gill81}; the probability law
here is the layered distribution above, not their independent fair-bit
oracle.
\end{proof}

\subsection{Consequences}

\begin{corollary}[Optimality of the randomized exponent among
relativizing arguments]\label[corollary]{cor:oracle-hd}
Fix integers $d\ge2$ and $c_0\ge1$. There is an oracle $A$ relative to which
$\hyp^A$ holds with threshold $\lceil c_0\lambda(n)^d\rceil$, $\mathrm{NP}^A=\mathrm{coNP}^A$,
and
\[
  \mathrm{NP}^A\not\subseteq
  \mathrm{BPTIME}^A\bigl(2^{\delta\lambda(n)^d}\bigr)
  \quad\text{for every fixed } 0<\delta<c_0 .
\]
In particular $\mathrm{NP}^A\not\subseteq\mathrm{BPP}^A$.
\end{corollary}

\begin{proof}
Apply Theorem~\ref{thm:oracle} with
$K(n)=\max\{2,\lceil c_0\Llog(n)^d\rceil\}$. For rational $0<\delta<c_0$, the
bound $T(n)=2^{\lceil\delta\lambda(n)^d\rceil}$ is time-constructible and
satisfies $T(n)2^{-K(n)}\le 2^{1-(c_0-\delta)\lambda(n)^d}\to0$.
Irrational $\delta$ follow by monotonicity. For each positive $\delta$ this clock eventually dominates every
polynomial, since $d\ge2$; a finite change ensures $T(n)\ge n$ at
all lengths.
\end{proof}

\begin{remark}[Scope of the relativized limit]
For this same oracle, reconstruction and isolation put
$\NP^A\subseteq\RTIME^A(2^{O(\Llog(n)^d)})$. The lower bound excludes
$\BPTIME^A(2^{O(\Llog(n)^{d-\varepsilon})})$ for every fixed
$0<\varepsilon<d$, and hence fixes the logarithmic power $d$ among
relativizing bounds. The leading constant is measured in the chosen
complete-set encoding.
More generally, for every admissible $K(n)=\omega(\log n)$,
\cref{thm:oracle} gives an oracle with a $K(n)$-comparator but
$\NP^A\not\subseteq\mathrm{BPP}^A$; thus the logarithmic-arity
conclusion $\NP=\mathrm{RP}$ of~\cite{Siv99} has no general
relativizing extension to this superlogarithmic range.
The oracle depends on $K$. Finally, $\NP^A=\coNP^A$ implies
$\PH^A=\NP^A$: the construction obstructs faster randomized SAT
algorithms while remaining compatible with the proposed second-level
collapse \eqref{eq:target}.
\end{remark}

\begin{remark}[Why the same table also collapses NP and coNP]
Syntactically accepting companions turn the rare answer $1^{K(m)}$
into a certificate of non-membership. The independent entries make
such certificates plentiful enough for nondeterministic search, while
the collision estimate limits what a short randomized computation can
learn. A different comparator distribution would be required for a
construction separating the second and third levels under the same
comparability hypothesis.
\end{remark}

\section{Randomized comparators}\label{sec:randomized}

\subsection{Error below a uniform guess}

A \emph{randomized comparator} for a language $A$ is a probabilistic
polynomial-time algorithm $g(\tau;r)$ that outputs a $t$-bit vector on
every encoded $t$-tuple $\tau$. Its \emph{error} on $\tau$ is
\[
  \mathrm{err}_g(\tau)=\Pr_r[g(\tau;r)=\chi_A(\tau)].
\]
Thus an error is the one forbidden output, not disagreement with a
specified canonical output. The running time is polynomial in the total
tuple encoding length on every coin sequence. Use an explicit tuple
encoding whose length $\ell$ is at least both the arity $t$ and the maximum
individual-instance length $n$. For each fixed algorithm, pad its random
tape to a fixed, polynomial-time computable polynomial bound in $\ell$;
unused bits are ignored. Each repetition below uses a fresh independent
tape unless coins are explicitly fixed as advice.

\begin{remark}[Why the threshold is $2^{-t}$]
\label[remark]{rem:uniform-guess}
The algorithm that ignores its input and outputs a uniform vector in
$\bits^t$ has error exactly $2^{-t}$ on every $t$-tuple, for every
language. Allowing error at least this baseline is therefore vacuous.
In the growing-arity setting considered here, any fixed positive error
allowance is also vacuous, after handling the finitely many exceptional
small tuples directly; this observation does not apply to arbitrary
fixed-arity, sub-baseline error requirements. The elimination formulation
and the logarithmic-arity error-reduction argument in the thesis already
use this baseline~\cite[Section~4.1 and Lemma~5.4.6]{BD13}.
\end{remark}

Throughout this section, $n$ is the maximum individual length, as in
\cref{def:membership-comparability}; $\ell$ is the total tuple encoding
length. Fixed constants in arity and exponent bounds can be enlarged to
integers.

\begin{definition}[Weak randomized comparability,
$\hyp^{\mathrm{wr}}$]\label[definition]{def:wr}
Fix an integer $d\ge1$. There are integers $c_0,e\ge1$ and a randomized
comparator $g$ for $\SAT$ such that every tuple of arity
$t\ge c_0\Llog(n)^d$ satisfies
\[
  \mathrm{err}_g(\tau)\le(1-\gamma(n))2^{-t},
  \qquad \gamma(n)=2^{-\lceil e\Llog(n)^d\rceil}.
\]
A larger guaranteed relative saving also satisfies the hypothesis by
using this lower bound for the saving.
\end{definition}

A deterministic comparator satisfies \cref{def:wr} with zero error.
Pseudodeterministic comparators with a correct canonical output and
zero-error comparators with constant nonfailure probability are
sufficient examples for the stronger, inverse-polynomial relative-gap
regime of \cref{def:relative-saving}: $O(\ell+1)$ independent repetitions,
using whole-vector majority or the first nonfailure output, respectively,
reduce the true-vector error below $2^{-\ell-1}\le2^{-t-1}$.
An expected-polynomial-time zero-error procedure can first be truncated
at twice its expected-time bound to obtain the required constant
nonfailure probability and a polynomial clock on every tape.

\subsection{Amplification inside the reconstruction clock}

\begin{lemma}[Mode amplification]\label[lemma]{lem:mode}
Let $g$ be a randomized comparator, and let $\tau$ be a $t$-tuple with
$\mathrm{err}_g(\tau)\le(1-\gamma)2^{-t}$, where $0<\gamma\le1$.
For $0<\eta<1$, draw
\begin{equation}
 M=\left\lceil10\cdot2^t\gamma^{-2}\ln(2/\eta)\right\rceil
 \label{eq:mode-samples}
\end{equation}
independent samples of $g(\tau;\cdot)$ and return their empirical mode,
breaking ties in a fixed order. The returned vector is the true vector
with probability at most $\eta$. The sampling cost is $M$ evaluations
of $g$; counting the sampled vectors adds
$M\,\mathrm{poly}(t,\log(M+1))$ deterministic overhead.
\end{lemma}

\begin{proof}
Put $B=2^t$, let $p_b$ be the probability of output $b$, and let $X_b$
count that output among the $M$ samples. Some $b^*$ has $p_{b^*}\ge1/B$;
since $p_\chi\le(1-\gamma)/B<1/B$, this pattern is not the true vector
$\chi=\chi_A(\tau)$. Set
\[
 \theta=(1-\gamma/2)M/B.
\]
If the empirical mode equals $\chi$, then $X_\chi\ge X_{b^*}$, so at
least one of $X_\chi\ge\theta$ and $X_{b^*}\le\theta$ occurs. No union
bound over all output patterns is needed.

The mean $\mu^*$ of $X_{b^*}$ is at least $M/B$, and
$\theta\le(1-\gamma/2)\mu^*$. Thus
\[
 \Pr[X_{b^*}\le\theta]
 \le\exp(-\gamma^2\mu^*/8)
 \le\exp(-\gamma^2M/(8B)).
\]
For the other event, $X_\chi$ is stochastically dominated by a binomial
variable of mean $\mu_0=(1-\gamma)M/B$. With
$a=\theta-\mu_0=\gamma M/(2B)$, Bernstein's inequality gives
\[
 \Pr[X_\chi\ge\theta]
 \le\exp\left(-\frac{a^2}{2\mu_0+2a/3}\right)
 \le\exp(-\gamma^2M/(10B)).
\]
The case $\gamma=1$ also follows directly because $X_\chi=0$.
Both displayed bounds are at most $\eta/2$ for
\eqref{eq:mode-samples}. Sorting the sampled strings and counting equal
runs implements the mode with the stated polynomial overhead in their
bit lengths. In the algorithms below, $\gamma$ is dyadic; replacing
$\ln(2/\eta)$ by $\lceil\log_2(2/\eta)\rceil$ makes the repetition
count an explicit integer upper bound with the same asymptotic cost.
\end{proof}

\begin{proposition}[Randomized reconstruction]\label[proposition]{prop:rand-reconstruction}
Under $\hyp^{\mathrm{wr}}$ there is a randomized search algorithm $U''$
that, on every circuit $C$ of encoding length $N$ and every coin
sequence, halts within $2^{O(\Llog(N)^d)}$ steps and returns an assignment
or $\bot$. Every returned assignment satisfies $C$, on every tape. If
$C$ is uniquely satisfiable, $U''$ returns its satisfying assignment
with probability at least $2/3$.
\end{proposition}

\begin{proof}
Use the field, auxiliary formulas, and parameters $m,K,q$ of
\cref{prop:unique}, with the threshold constant from
\cref{def:wr}. The constructed formulas have length at most $N^b$ as
in \eqref{eq:constructed-size}. There are at most
$q\binom mK\le q^2$ comparisons, all of arity $K$. Replace each by
\cref{lem:mode}, with
\[
 \gamma=\gamma(N^b),\qquad \eta=1/(3q^2).
\]
This is a valid lower bound on the relative saving for every such call,
because $\gamma$ is nonincreasing and the maximum individual length is
at most $N^b$.

\emph{Clock on every tape.} Every comparator call returns one $K$-bit
pattern for its coordinate subset. Whether or not that pattern is the
true vector, the surviving set shatters no $K$-subset. Hence
\eqref{eq:shortlist}, and all the reconstruction and verification loop
bounds in \cref{prop:unique}, hold on every coin sequence. Empty pair
lists return $\bot$ directly. An amplified comparison uses
\[
 2^K\gamma(N^b)^{-2}\,O(m)
\]
samples, up to harmless integer rounding. Since
$K=O(\Llog(N)^d)$ and
\[
 \gamma(N^b)^{-2}
 \le 2^{2e(b+1)^d\Llog(N)^d+2},
\]
the sample cost, the mode-counting overhead, and all $q^2$ calls together
remain within $2^{O(\Llog(N)^d)}$. Finitely many smaller inputs and
ground circuits are handled directly as before.

\emph{Soundness.} Each reconstructed candidate is checked in $C$ before
it is returned. Thus no coin sequence creates a false witness, and
unsatisfiable circuits always return $\bot$.

\emph{Success on unique inputs.} A union bound makes every amplified
comparison correct with probability at least $2/3$. Conditional on
this event and uniqueness, the true evaluation $P_a(u)$ survives at
every field point. The polynomial agrees with all $q$ true point--value
pairs and is recovered by \cref{lem:sudan}, exactly as in
\cref{prop:unique}. It passes the final direct check.
\end{proof}

\begin{corollary}[Randomized SAT algorithm]\label[corollary]{cor:rand-sat}
Under $\hyp^{\mathrm{wr}}$,
\[
 \SAT,\CSAT\in\RTIME\bigl(2^{O(\Llog(n)^d)}\bigr).
\]
For $d=1$, this gives $\NP=\class{RP}$.
\end{corollary}

\begin{proof}
Use the isolation construction of \cref{lem:isolation} with $U''$ in
place of the deterministic promise solver, and accept only upon finding
a checked assignment. An unsatisfiable input is never accepted. For a
satisfiable input, a uniform isolation seed makes some prefix restriction
unique with probability at least $3/16$. Conditional on that seed,
choose the first such prefix only for the analysis; its independent
reconstruction coins succeed with probability at least $2/3$. One
trial therefore accepts with probability at least $1/8$. A fixed
constant number of independent trials raises this to $2/3$, preserving
the clock. For $d=1$ it is polynomial, and the usual reduction to SAT
gives $\NP=\class{RP}$.
\end{proof}

\subsection{Consequences of weak randomized comparability}

\begin{theorem}[Weak randomized comparability]\label{thm:wr}
Assume $\hyp^{\mathrm{wr}}$. Then the following hold.
\begin{enumerate}[label=(\roman*),beginpenalty=10000]
 \item For every fixed $k\ge1$,
 \[
  \Sigma_k^p\cup\Pi_k^p\subseteq
  \BPTIME\bigl(2^{O(\Llog(n)^{d^k})}\bigr).
 \]
 In particular, $\PH\subseteq\BPQP$; the exponent here may depend on
 the fixed hierarchy level.
 \item $\QH=\BPQP$.
 \item $\NEXP=\REXP$ and $\EXPH=\BPEXP$.
 \item The fixed, unconditionally defined language $G$ from
 \cref{cor:unconditional-qh-hard} satisfies
 \[
  G\in\BPTIME\!\left(
     2^{O((\Llog(n)\log\Llog(n))^{d^5})}\right)
  \setminus(\NP/\poly\cup\coNP/\poly).
 \]
 Consequently,
 \begin{equation}
  \mathrm{BPP}\subsetneq
  \BPTIME\!\left(2^{O((\Llog(n)\log\Llog(n))^{d^5})}\right).
  \label{eq:weak-time-separation}
 \end{equation}
 This larger time class is contained in $\BPQP$, and
 $\BPQP\not\subseteq\NP/\poly\cup\coNP/\poly$.
 \item $\BPEXP\not\subseteq\NP/\poly\cup\coNP/\poly$.
\end{enumerate}
\end{theorem}

\begin{proof}
(i) For $k=1$, use \cref{cor:rand-sat} and complementation. Suppose the
claim holds at level $k$, and write a language of $\Sigma_{k+1}^p$ as
\[
 x\in L\iff\exists u\in\bits^{p(n)}\ B(x,u),\qquad B\in\Pi_k^p.
\]
The common-tape argument of \cref{lem:second} works for this predicate
as well. Its bounded-error algorithm has time
$2^{O(\Llog(n)^{d^k})}$ on the polynomial-length pair $(x,u)$.
Amplify its error to at most $2^{-p(n)-6}$. For a fixed $x$, a union bound
over all $u$ shows that a shared tape $r$ makes the amplified computation
correct for all $u$ with probability at least $1-2^{-6}$. This tape can
be quasipolynomially long; it is sampled at runtime, not supplied as
polynomial advice.

For fixed $(x,r)$, compile $u\mapsto\widehat B(x,u;r)$ into an explicit
circuit of size
\[
 S=2^{O(\Llog(n)^{d^k})}.
\]
Run \cref{cor:rand-sat} on that circuit, amplifying its error to at most
$2^{-6}$. On a good tape the circuit is satisfiable exactly when
$x\in L$. The total error is bounded by the sum of the two failure
probabilities, and the total time is
\[
 S^{O(1)}+2^{O(\Llog(S)^d)}
 =2^{O(\Llog(n)^{d^{k+1}})}.
\]
Complementation gives $\Pi_{k+1}^p$. This is an induction for fixed $k$;
its constants and machines may depend on $k$.

(ii) The proof of \cref{lem:qp-block,thm:qh} requires only a randomized
SAT algorithm whose clock is closed under fixed quasipolynomial
compositions. Substitute \cref{cor:rand-sat} for
\cref{cor:random-sat}. The reverse covering containment is unchanged
and unconditional.

(iii) The proof of $\NEXP\subseteq\REXP$ in \cref{thm:exp} uses only
the one-sided SAT algorithm; the reverse containment is unconditional.
For $\EXPH\subseteq\BPEXP$, repeat the padding construction in
\cref{lem:padding} for a fixed level $k$, now using part (i) in place
of \cref{thm:ph}. A language in $\Sigma_k\mathrm{EXP}$ with resource
bound $2^{n^c}$ pads to a $\Sigma_k^p$ language on length
$N\le2^{n^c+1}$, after a finite-length normalization. Its randomized
clock is $2^{O(n^{c d^k})}$, including the construction cost. The
reverse containment is the unconditional covering argument used in
\cref{cor:exph}.

(iv) Write $\ell(n)=\Llog(n)\max\{1,\lceil\log_2\Llog(n)\rceil\}$.
The unconditional lemmas in \cref{sec:nd-hard-extension} give
$G(x)=H(x_1\cdots x_{\ell(|x|)})$ above a fixed cutoff, with
$H\in\Sigma_5\mathrm{TIME}(2^{O(m)})$ and
$G\notin\NP/\poly\cup\coNP/\poly$. Pad the $m$-bit input of $H$
to length $N=2^{O(m)}$ and use part (i) with $k=5$. This gives time
$2^{O(m^{d^5})}$ for $H$ and hence $2^{O(\ell(n)^{d^5})}$ for $G$.
Thus $G\in\BPQP$ but $G\notin\mathrm{BPP}$ by
$\mathrm{BPP}\subseteq\Pclass/\poly$~\cite{Adl78}. The clock in
\eqref{eq:weak-time-separation} dominates every polynomial, so the
inclusion of BPP is also valid. The circuit exclusion is inherited from
the unconditional construction; its randomized upper bound is the step
using $\hyp^{\mathrm{wr}}$.

(v) Part (iii) places the same unconditional $H$ in $\BPEXP$. Its
near-maximum signed existential circuit lower bound excludes both
polynomial-advice classes.
\end{proof}

\begin{remark}[Structural and advice distinctions under the weak hypothesis]
\label[remark]{rem:weak-scope}
For $d>1$, whether $\Hwr$ implies $\NP\subseteq\coNP/\poly$
remains unresolved here. \Cref{subsec:certificate-boundary} identifies
the obstruction in the present method: limited independence still gives
polynomial-size table descriptions and advised certified SAT search,
but a heavy-output certificate can have quasipolynomial length. The
positive-relation theorem then lacks a polynomial-time $\NP/\poly$
verifier. Consequently the structural collapse $\PH=\Sclass^{\NP}$,
the level-independent exponent $d^2$, and the polynomial-length
symmetric simulation of the whole PH have not been obtained in this
regime. Under an inverse-polynomial relative saving,
\cref{thm:relative-gap-transfer} establishes these conclusions.
At $d=1$, \cref{cor:rand-sat} already gives $\NP=\class{RP}$ and
the structural consequences follow.

The exponent $d^5$ in \cref{thm:wr}(iv) comes from the fifth-level
signed existential diagonal language. For separation from BPP alone,
a deterministic-circuit diagonalization suffices. Its canonical-table
formula has the four blocks
$\exists T\,\forall(D,T')\,\exists(y,D')\,\forall y'$: ordinary circuit
disagreement needs one input, and agreement with an earlier table needs
one universal input. The same counting and prefix projection give an
unconditional $\Sigma_4\mathrm{TIME}(2^{O(m)})$ hard language and a
projected language outside $\Pclass/\poly$. Applying \cref{thm:wr}(i)
with $k=4$ gives the sharper BPP-only separation
\[
 \mathrm{BPP}\subsetneq
 \BPTIME\!\left(2^{O((\Llog(n)\log\Llog(n))^{d^4})}\right).
\]
The fifth-level version is retained in the theorem for its exclusion of
$\NP/\poly\cup\coNP/\poly$, not just $\Pclass/\poly$; the underlying
diagonal method is the same Kannan-style construction~\cite{kannan82}.

The one-sided randomized analogue of \eqref{eq:binary-advice} does
survive. More generally, if a designated advice $a_n$ and a
polynomial-time verifier $R$ give
$x\in L\iff\exists z\ R(x,a_n,z)$ with polynomial witness length,
compile $R(x,a_n,\cdot)$ into a circuit and run \cref{cor:rand-sat}.
The only advice supplied is $a_n$; the SAT subroutine samples its own
isolation and reconstruction coins. Thus
\begin{equation}
 2\mc\subseteq
 \RTIME\bigl(2^{O(\Llog(n)^d)}\bigr)/O(n).
 \label{eq:randomized-binary-advice}
\end{equation}
Here $2\mc$ is the original deterministic binary-comparator class,
and correctness is required under its designated advice. The uniform
deterministic promise-unique algorithm is replaced by
\cref{prop:rand-reconstruction}; the proof gives
$\UEXP\subseteq\REXP$, not the equality $\UEXP=\EXP$.
\end{remark}

\begin{remark}[The relativized limit covers both randomized regimes]
\label[remark]{rem:randomized-oracle}
The comparator in \cref{sec:oracle} is deterministic, so the same oracle
satisfies $\hyp$, $\Hrel$, and $\Hwr$. The proofs of
\cref{prop:rand-reconstruction,cor:rand-sat} relativize, using oracle
circuits and the relativized complete set. Thus \cref{cor:oracle-hd}
limits improvements to the logarithmic power $d$ in the uniform
randomized SAT clock in either randomized regime. Its lower bound
concerns this SAT clock, rather than the PH exponents $d^k$ and $d^2$.
Relative to this oracle $\PH^A=\NP^A$, as before.
\end{remark}

\section{Relative-gap comparators via succinct sample tables}
\label{sec:relative-gap}

\subsection{Hypothesis and result}

Write $\Td(n)=2^{O(\Llog(n)^d)}$, with the union-over-constants
convention of \cref{sec:intro}. The randomized comparator is polynomially
clocked on every tape; the derived reconstruction and search procedures
have the longer clocks stated below. Advice is fixed per input length.
Correctness of an advised decision relation is required only on its
designated advice.

\begin{definition}[Inverse-polynomial relative saving]\label[definition]{def:relative-saving}
Fix an integer $d\ge1$. The hypothesis $\Hrel$ asserts that there are
fixed positive integers $a,c_0$ and a randomized polynomial-time algorithm
$g$ such that, on every $t$-tuple $\tau$ of SAT instances of maximum
individual length $n$, with $t\ge c_0\Llog(n)^d$,
\[
 \Pr[g(\tau)=\chi_{\SAT}(\tau)]
 \le \bigl(1-(n+2)^{-a}\bigr)2^{-t}.
\]
The output has exactly $t$ bits on every tape.
\end{definition}

The saving is relative to the uniform-guess error $2^{-t}$; its absolute
value $2^{-t}(n+2)^{-a}$ can be smaller than every inverse polynomial
when $t=\omega(\log n)$. The hypotheses satisfy
\[
 \hyp\ \Longrightarrow\ \Hrel\ \Longrightarrow\ \Hwr,
\]
since a deterministic comparator has zero error and
$(n+2)^{-a}\ge2^{-a\Llog(n)^d}$.
The stronger condition $\Hsr$ of polynomial-time error reduction to
$2^{-p(\ell)}$ for every fixed polynomial $p$, with a common arity
threshold, satisfies $\Hsr\Rightarrow\Hrel$: take $p(\ell)=\ell+1$
and use $2^{-\ell-1}\le2^{-t-1}$.

\Needspace{8\baselineskip}
\begin{theorem}[Relative-gap transfer]\label{thm:relative-gap-transfer}
Under $\Hrel$ the following statements hold.
\begin{enumerate}[label=(\roman*)]
\item $\SAT\in\NP/\poly\cap\coNP/\poly$, and
      $\PH=\class{S}_2^{\NP}$.
\item $\PH\subseteq\BPTIME(2^{O(\Llog(n)^{d^2})})$.
\item $\PH\subseteq\class{S}_2^p[\Td]$ and
      $\PH\subseteq\DTIME(\Td)/\poly$. In the first containment
      both certificates have polynomial length and the predicate is
      deterministic and oracle-free.
\item Promise Unique-Circuit-SAT has deterministic search in time
      $\Td$ with polynomial advice, and certified satisfiability search
      has such advice and clock with soundness under every candidate
      advice string.
\end{enumerate}
Together with \cref{cor:relative-consequences}, this gives every
conclusion of \cref{thm:main}, with the same exponents, except the
uniform deterministic promise-unique algorithm in item~(i) and the
equality $\UEXP=\EXP$ in item~(iv). Uniform promise-unique search has
the one-sided randomized guarantee of \cref{prop:rand-reconstruction}.
The deterministic binary-comparator bound \eqref{eq:binary-advice}
is replaced by
\[
 2\mc\subseteq\DTIME(\Td)/\poly;
\]
its one-sided randomized version \eqref{eq:randomized-binary-advice}
retains ordinary linear advice.
\end{theorem}

The sample-table construction uses only limited independence; both
its moment estimate and its polynomial-evaluation implementation are
given below. Appendix~\ref{app:positive-abg} supplies the full
positive-relation advice transfer. The randomized-comparability definition in
Beigi--Etesami--Gohari instead requires inverse-polynomial \emph{absolute}
saving over $2^{-t}$~\cite[Definition~1 and Lemma~3]{BEG17}; that
condition restricts their parameter to logarithmic arity. The distinction
from relative saving is essential here.

\subsection{Limited independence and local table evaluation}

Only an upper-tail bound for each tuple is needed. We use the classical
limited-independence moment method of Schmidt, Siegel, and
Srinivasan~\cite{SSS95}; the following elementary form, including the
constants used below, is proved here.

\begin{lemma}[An even-moment tail bound]\label[lemma]{lem:limited-tail}
Let $Y_1,\ldots,Y_J\in\{0,1\}$ be $r$-wise independent, where $r\ge2$
is even and $r\le J$. Put $S=\sum_iY_i$ and $\mu=\E S$.
If $0<\gamma\le1$, $\mu\le(1-\gamma)R$, and
$R\ge32r\gamma^{-2}$, then
\[
 \Pr[S\ge R]\le2^{-r}.
\]
\end{lemma}

\begin{proof}
The $r$th centered moment agrees with that of fully independent
Bernoulli variables having the same marginals: every monomial in its
expansion involves at most $r$ distinct variables. For independent
variables put $p_i=\E Y_i$ and $Z_i=Y_i-p_i$. Then $\E Z_i=0$ and,
for every integer $j\ge2$,
$|\E Z_i^j|\le\E|Z_i|^j\le\E Z_i^2\le p_i$.
Expanding the moment and majorizing the coefficients by nonnegative
power series gives, for any $z>0$,
\begin{align*}
 \E(S-\mu)^r
 &\le r!\,[z^r]\prod_i\bigl(1+p_i(e^z-1-z)\bigr)\\
 &\le r!\,[z^r]\exp\bigl(\mu(e^z-1-z)\bigr)\\
 &\le r!z^{-r}\exp\bigl(\mu(e^z-1-z)\bigr).
\end{align*}
In the last line $z$ is a positive evaluation point; coefficient
majorization justifies the preceding two lines. Choose
$z=\sqrt{r/(2R)}\le1$. Using $e^z-1-z\le z^2$, $\mu\le R$, and
$r!\le r^r$, we obtain
\[
 \E(S-\mu)^r\le(2erR)^{r/2}.
\]
Since $r$ is even and $S\ge R$ implies $|S-\mu|\ge\gamma R$,
Markov's inequality now gives
\[
 \Pr[S\ge R]
 \le\left(\frac{2er}{\gamma^2R}\right)^{r/2}
 \le(1/4)^{r/2}=2^{-r}.
\]
\end{proof}

\subsection{Certifiable heavy outputs}

The construction applies to any language $A$. Let $k(n)\ge1$ be
polynomially bounded, nondecreasing, and polynomial-time computable.
Suppose a polynomially clocked randomized algorithm $g$, on every
$k(n)$-tuple $\tau$ of strings of length at most $n$, satisfies
\begin{equation}
 \Pr[g(\tau)=\chi_A(\tau)]\le(1-\gamma_n)2^{-k(n)},
 \qquad \gamma_n=(n+2)^{-a},
 \label{eq:relative-gap}
\end{equation}
for a fixed integer $a\ge1$.

\begin{lemma}[A locally evaluable limited-independence table]
\label[lemma]{lem:relative-table}
Under \eqref{eq:relative-gap}, there are polynomial-length advice
strings $\sigma_n$ and a deterministic polynomial-time tape evaluator
$\mathcal R(n,\sigma_n,\tau,i)$ with the following property. There are
$J=2^{k(n)}R$ sample positions, for a polynomially bounded integer $R$,
and simultaneously for every permitted tuple $\tau$,
\begin{align}
 \#\{i<J:g(\tau;\mathcal R(n,\sigma_n,\tau,i))
                    =\chi_A(\tau)\}&<R,\notag\\
 \max_{b\in\bits^{k(n)}}
 \#\{i<J:g(\tau;\mathcal R(n,\sigma_n,\tau,i))=b\}&\ge R.
 \label{eq:relative-pigeonhole}
\end{align}
The evaluator is polynomially clocked on all candidate advice strings.
\end{lemma}

\begin{proof}
Write $k=k(n)$. Encode a string of length at most $n$ by the nonzero
$(n+1)$-bit string $0^{n-|x|}1x$. Thus tuple codes have length
$b_n=(n+1)k$, and there are at most $2^{b_n}$ valid codes. Let $c_n\ge1$
be a power-of-two polynomial upper bound on the number of comparator
coins, with unused coins padded. Set
\begin{equation}
 r_n=2(b_n+2),\qquad
 R=2^{\lceil\log_2(32r_n\gamma_n^{-2})\rceil},\qquad
 J=2^kR.
 \label{eq:limited-parameters}
\end{equation}
View the full table as $2^{b_n}$ rows, $J$ blocks per row, and $c_n$
bits per block. Its bit length is
\[
 M_n=2^{b_n}Jc_n,
 \qquad s_n=\log_2M_n=\poly(n).
\]
Generate these bits with exact $q_n$-wise independence, where
$q_n=r_nc_n$. To do so, fix a monic irreducible $f_n$ of degree $s_n$
over $\F_2$, identify the $M_n$ bit addresses with the elements of
$\F_{2^{s_n}}=\F_2[Z]/(f_n)$, and choose independently uniform field
coefficients $a_0,\ldots,a_{q_n-1}$. The bit at address $u$ is one
fixed nonzero binary coordinate of
\[
 P(u)=\sum_{j=0}^{q_n-1}a_ju^j.
\]
For any at most $q_n$ distinct addresses, the Vandermonde evaluation
map is surjective. Their field values, and hence the selected output
bits, are independent and uniform. This also covers any comparator
that reads its tape adaptively: each block is a full uniform tape.

Fix one tuple $\tau$, and let $Y_i$ indicate that block $i$ produces
$\chi_A(\tau)$. Any $r_n$ different blocks use $r_nc_n=q_n$ distinct
table coordinates, so these indicators are $r_n$-wise independent.
Their sum $N_\tau$ has mean at most $(1-\gamma_n)R$. Therefore
\cref{lem:limited-tail} gives
\[
 \Pr[N_\tau\ge R]\le2^{-r_n},\qquad
 \Pr[\exists\text{ valid }\tau:N_\tau\ge R]
 \le2^{b_n-r_n}=2^{-b_n-4}<1.
\]
No independence between different rows is required. Some coefficient
choice therefore has the desired strict bound for every tuple.

Advise that coefficient list and $f_n$. The total length is
$q_ns_n+s_n+1=\poly(n)$. A requested bit is obtained by Horner
evaluation; reading one comparator tape requires $c_n$ such evaluations.
All arithmetic is on $s_n$-bit field elements, so local evaluation is
polynomial-time without expanding the table. On arbitrary candidate
advice, require only the fixed lengths and a monic degree-$s_n$ modulus:
arithmetic in the quotient ring is still polynomial-time even when the
modulus is reducible. Malformed descriptions use a fixed all-zero
table. These conventions preserve the all-advice clock; the independence
argument uses the designated irreducible modulus.

Finally, on every table, $J=2^kR$ outputs occupy $2^k$ patterns, so at
least one pattern occurs $R$ times. On the designated table that pattern
cannot be the true vector, proving \eqref{eq:relative-pigeonhole}.
\end{proof}

\begin{corollary}[A total comparison relation in NP/poly]
\label[corollary]{cor:relative-np-relation}
There is a relation $W$ such that, for every permitted tuple $\tau$,
\[
 \exists b\ W(\tau,b),\qquad
 W(\tau,b)\Longrightarrow b\ne\chi_A(\tau),
\]
and membership in $W$ is verifiable nondeterministically in polynomial
time with polynomial advice.
There is also a deterministic advised comparator for these tuples
running in time $2^{k(n)}\poly(n)$, with polynomial advice.
\end{corollary}

\begin{proof}
Under the designated $\sigma_n$, define $W(\tau,b)$ to mean that
$b$ occurs in at least $R$ of the $J$ sampled outputs. A witness lists
$R$ distinct sample indices. Its length is $O(R\log J)=\poly(n)$,
and the verifier evaluates the corresponding tapes and runs $g$.
Totality and exclusion of the true vector follow from
\cref{lem:relative-table}. The parameter $n$ may be supplied in unary; to put
the relation in the ordinary input-length advice convention, bundle
$\sigma_j$ for all $j$ up to the relation-input length. This remains
polynomial advice.

For deterministic comparison, enumerate all $J$ outputs and return their
empirical mode, breaking ties in a fixed order. Its count is at least
$R$, whereas the true vector's count is below $R$. Enumeration and
sorting cost $2^{k(n)}\poly(n)$; $\log J$ is polynomial. This procedure
has its stated clock and always returns a vector even on incorrect
candidate advice, using a fixed default parsing convention. For a tuple
of arity larger than $k(n)$, where $n$ is the maximum length in the whole
tuple, apply the level-$n$ procedure to its first $k(n)$ entries and append
zeros. The prefix entries have length at most $n$, so its incorrect
prefix vector makes the entire returned vector incorrect.
\end{proof}

The relation in \cref{cor:relative-np-relation} is not asserted to be in
$\Pclass/\poly$. A heavy vector has a short certificate, whereas
finding it by the construction can take $2^{k(n)}\poly(n)$ time.
This distinction is what permits the next step without asserting
polynomial-size SAT decision circuits.

\subsection{The positive form of the advice theorem}

\begin{lemma}[Positive-query Amir--Beigel--Gasarch transfer]
\label[lemma]{lem:positive-abg}
Let $k(n)\ge1$ be polynomially bounded. Suppose a relation $W$ is total
on $k(n)$-tuples of $n$-bit inputs and accepts only vectors different from
the true $A$-membership vector. If $W\in\NP/\poly$, then
\[
 A\in\NP/\poly\cap\coNP/\poly.
\]
The two verifiers may use a common polynomial-length advice sequence.
\end{lemma}

\begin{proof}
Appendix~\ref{app:positive-abg} gives a self-contained construction.
The anchor-cover branch uses one positive $W$-certificate. In the
residual branch, a fixed list of $M=16(n+1)+1$ tuples has a strict
majority of good positions for every input. The NP verifier selects
exactly $(M+1)/2$ distinct positions and supplies one positive
$W$-certificate for each. The appendix proves the simultaneous
majority guarantee and the soundness of this verifier for both
requested membership bits.
\end{proof}

\begin{corollary}[Polynomial arity with an inverse-polynomial relative gap]\label[corollary]{cor:relative-advice}
Under \eqref{eq:relative-gap},
\[
 A\in\NP/\poly\cap\coNP/\poly.
\]
In particular, $\Hrel$ implies this containment for SAT.
\end{corollary}

\begin{proof}
Restrict \cref{cor:relative-np-relation} to equal-length inputs and apply
\cref{lem:positive-abg}. For SAT take
$k(n)=\lceil c_0\Llog(n)^d\rceil$. Every tuple of inputs of length
at most $n$ is covered, and its relative saving is at least
$(n+2)^{-a}$ by monotonicity of that lower bound.
\end{proof}

\subsection{Recovering the clocks and the symmetric simulation}

The reconstruction and isolation arguments in
\cref{sec:unique,sec:isolation,sec:randomized} give the following interfaces:
\begin{enumerate}[label=(\alph*)]
\item a weak randomized comparator gives one-sided randomized SAT search
      in time $\Td$;
\item a deterministic valid comparator evaluated in time $\Td$ on the
      polynomial-size auxiliary tuples still gives deterministic
      promise-unique search in time $\Td$;
\item the reconstruction clock is valid for arbitrary excluded patterns:
      each coordinate subset forbids one pattern, so Sauer--Shelah bounds
      the surviving list even when the exclusion is semantically wrong;
\item every returned candidate is checked directly in the input circuit.
\end{enumerate}
Here multiplication of a fixed number of $\Td$ factors, and polynomial
blowup in the argument of $\Td$, preserve its exponent $d$.
These interfaces are the parameter-explicit reconstruction of
\cref{prop:unique,prop:rand-reconstruction} and the randomized SAT
algorithm of \cref{cor:rand-sat}, based on Sivakumar's method~\cite{Siv99}.

\begin{proof}[Proof of \cref{thm:relative-gap-transfer}]
\Cref{cor:relative-advice} gives $\coNP\subseteq\NP/\poly$.
Apply \cref{thm:yap}~\cite{CCHO05} to obtain $\PH=\Sclass^{\NP}$.
Also $\Hrel\Rightarrow\Hwr$, so \cref{cor:rand-sat} supplies the
uniform randomized SAT algorithm. Substitute these two inputs into
\cref{lem:second} and the proof of \cref{thm:ph}. They give the
level-independent clock $2^{O(\Llog(n)^{d^2})}$ for all of PH.

For deterministic advised reconstruction, replace each comparison on
an auxiliary tuple of maximum length $v$ by the empirical-mode
comparator from \cref{cor:relative-np-relation}, using its advice $\sigma_v$.
Since $v=N^{O(1)}$ and $k(v)=O(\Llog(N)^d)$, each replacement costs
$2^{O(\Llog(N)^d)}$. There are at most $2^{O(\Llog(N)^d)}$
comparisons, so the total clock has the same exponent. Bundle the
$\sigma_v$ for all lengths up to a fixed polynomial bound in $N$;
the resulting advice is polynomially long. Correct advice makes all
comparisons valid and supplies promise-unique completeness.

For certified search, apply the isolation-list construction in
\cref{lem:certified-search}.
The list for length $N$ has $O(N)$ blocks of length $3N+1$; all restricted
circuits have length at most a fixed polynomial $S(N)$. Include table
advice for every auxiliary maximum length reached by reconstruction on
circuits up to $S(N)$. If $B(v)$ bounds those lengths, the added advice
has length at most
\[
 \sum_{j\le B(S(N))}|\sigma_j|=\poly(N).
\]
Resetting and reusing a description at a length introduces no further
advice cost. There are only polynomially many isolated circuits.

On arbitrary candidate advice, every comparison still returns a pattern,
so the reconstruction list and time bounds remain valid. Any witness is
checked directly in the original circuit; thus incorrect advice can
cause failure but not a false positive. Under designated table advice,
the reconstruction is correct on every uniquely satisfiable restriction,
and a designated isolation list hits one for every satisfiable circuit
of the input length. This establishes certified search with the stated
clock, polynomial advice, and all-advice soundness.

Use that search as $F$ in \cref{lem:two-bundles} and the proof of
\cref{thm:s2-simulation}. Each prover includes seed lists and the table
descriptions for all polynomially bounded query lengths. The OR of the
two verified search outcomes equals CircuitSAT whenever one bundle is
designated. Induction on the adaptive query sequence therefore preserves
the winning-certificate guarantees of $\class{S}_2^{\NP}$. Both bundles
have polynomial length, and the resulting deterministic predicate runs
in time $\Td$ on all bundles. This proves
$\PH\subseteq\class{S}_2^p[\Td]$.

Finally, \cref{lem:ph-np-advice} makes NP/poly closed under
complementation and polynomially bounded projections, yielding
$\PH\subseteq\NP/\poly$. Compile each advised NP verifier into a
polynomial-size circuit with its advice and input hardwired. Running
certified search with its designated table-and-isolation advice decides
that language deterministically in time $\Td$ with polynomial advice.
This proves the remaining containment.
\end{proof}

\begin{corollary}[Time and advice consequences]\label[corollary]{cor:relative-consequences}
Under $\Hrel$,
\begin{gather*}
 \QH=\BPQP,\qquad
 \class{EXPH}=\class{BPEXP},\qquad
 \class{NEXP}=\class{REXP},\\
 \mathrm{BPP}\subsetneq
 \BPTIME\!\left(2^{O(\Llog(n)^{d^2+\varepsilon})}\right)
 \quad(\varepsilon>0\text{ fixed}).
\end{gather*}
The fixed signed-circuit diagonal language $G$ of
\cref{cor:unconditional-qh-hard} satisfies
\[
 G\in\BPTIME\!\left(2^{O((\Llog(n)\log\Llog(n))^{d^2})}\right)
 \setminus\bigl(\NP/\poly\cup\coNP/\poly\bigr).
\]
Its hardness is the unconditional diagonalization in
\cref{sec:nd-hard-extension}; the hypothesis supplies the randomized
clock. The same transfer places the fixed hard language $H$ in
$\BPTIME(2^{O(n^{d^2})})$ and yields
\[
 \PH/\poly=\NP/\poly=\coNP/\poly,\qquad
 \PH\subsetneq\QH=\BPQP,\qquad
 \BPQP\not\subseteq\PH/\poly.
\]
The fixed-polynomial circuit lower bounds of \cref{rem:kannan-clock}
also retain their exact exponent $d^2$.
\end{corollary}

\begin{proof}
Apply \cref{thm:relative-gap-transfer} in the proofs of
\cref{thm:qh,cor:exph}, the $\NEXP=\REXP$ part of \cref{thm:exp},
and \cref{thm:diagonal-randomized,cor:ph-advice,cor:ph-strict-qh}.
The proof of \cref{rem:kannan-clock} uses the same level-independent
PH simulation on its fixed projection language. This accounts for all
remaining assertions of \cref{thm:main} covered by the transfer.
The time separation follows exactly as in
\cref{cor:conditional-bpp-hierarchy}: use the unconditional $G$, the
level-independent PH simulation with \cref{lem:padding}, and
$\mathrm{BPP}\subseteq\Pclass/\poly$~\cite{Adl78}. For binary membership comparability,
compile the original linear-advice verifier and use randomized SAT search
to keep linear advice; deterministic certified search additionally stores
the polynomial-size table descriptions.
\end{proof}

\subsection{The certificate-length boundary}
\label{subsec:certificate-boundary}

The two regimes are separated by the length of the heavy-output
certificate, not by the length of the sample-table description. For
arity $k(n)$ and tuple-code length $b_n=(n+1)k(n)$, the construction
chooses the least power of two
\[
 R\ge32r_n\gamma_n^{-2},\qquad r_n=2(b_n+2),\qquad J=2^{k(n)}R.
\]
Hence
\begin{equation}
 R=\Theta\bigl((b_n+1)/\gamma_n^2\bigr),\qquad
 |\pi|=\Theta(R\log_2J)
       =\Theta\bigl(R(k(n)+\log_2R)\bigr),
 \label{eq:occurrence-boundary}
\end{equation}
where $\pi$ explicitly lists $R$ distinct occurrence indices. Each
index is checked by local table evaluation followed by one run of the
original polynomial-time comparator.

\paragraph{Inverse-polynomial relative advantage.}
If $\gamma_n^{-1}\le\poly(n)$, both $R$ and $|\pi|$ are polynomially
bounded. The heavy-output relation is in $\NP/\poly$, and
\cref{lem:positive-abg} converts it into two-sided nondeterministic
advice for SAT. This gives $\PH=\Sclass^{\NP}$ and all the transferred
conclusions of \cref{thm:relative-gap-transfer,cor:relative-consequences},
including PH exponent $d^2$ and symmetric predicate exponent $d$.
The uniform deterministic promise-unique solver and $\UEXP=\EXP$
remain specific to the deterministic hypothesis.

\paragraph{Quasipolynomially small relative advantage.}
If the available guarantee is only
$\gamma_n=2^{-\Theta(\Llog(n)^d)}$ with $d>1$, then for
$k(n)=\Theta(\Llog(n)^d)$ the threshold $R$ and the occurrence
certificate in \eqref{eq:occurrence-boundary} have length
$2^{\Theta(\Llog(n)^d)}$. Listing and checking those occurrences is
quasipolynomial, not polynomial. It therefore does not establish the
$\NP/\poly$ relation required by \cref{lem:positive-abg}.
The weaker guarantee still gives randomized SAT exponent $d$, PH
exponent $d^k$ at fixed level $k$, and the hierarchy equalities and time
separations of \cref{thm:wr}; whether it implies
$\NP\subseteq\coNP/\poly$ remains unresolved here.

The table advice stays polynomial. Indeed its independence order
$q_n=r_nc_n$ does not depend on $\gamma_n$, and its field degree is
\[
 s_n=b_n+k(n)+\log_2R+\log_2c_n=\poly(n).
\]
The same even-moment bound and polynomial-evaluation construction apply
with this larger $R$. Local evaluation depends polynomially on
$\log R$, so empirical-mode enumeration takes
$2^{k(n)}R\,\poly(n)=2^{O(\Llog(n)^d)}$ under $\Hwr$.
Advised reconstruction and isolation therefore still give certified
SAT search with polynomial advice and exponent $d$. They do not supply
a polynomial-time verifier for the heavy-output relation. The
certificate-length threshold explains the two hypotheses used in the
paper; it is a threshold of this method, rather than a converse theorem.
For $d=1$, the weak saving is already inverse-polynomial, so the two
regimes coincide up to constants.

For example, when $d=2$, a relative saving of $(n+2)^{-3}$ gives a
polynomial heavy-output certificate and randomized PH exponent $4$.
The weaker saving $2^{-\Theta(\lambda(n)^2)}$ retains randomized SAT
exponent $2$ but yields only exponent $2^k$ at the $k$th PH level by
\cref{thm:wr}, without the structural advice step.

Designated table descriptions are nonuniform, and an $\NP/\poly$
comparison relation need not have a polynomial-time output selector.
Thus the relative-gap transfer supplies neither $\SAT\in\Pclass/\poly$
nor the ordinary second-level collapse for $d>1$. The deterministic
oracle of \cref{sec:oracle} satisfies both randomized hypotheses, so
\cref{cor:oracle-hd} still limits improvements to the logarithmic power
$d$ in their randomized SAT clock.

\section{Summary and status}
\label{sec:summary}

Under $\hyp$, certified search gives
\[
 \PH\subseteq\Sclass^p\!\left[2^{O(\Llog(n)^d)}\right],\qquad
 \PH\subseteq\DTIME\!\left(2^{O(\Llog(n)^d)}\right)/\poly.
\]
The symmetric certificates have polynomial length. Two competing seed
bundles answer the simulated NP queries exactly as soon as one bundle
is designated, avoiding a completeness test for guessed advice.
The uniform randomized PH clock is $2^{O(\Llog(n)^{d^2})}$,
independent of the hierarchy level. Padding yields $\QH=\BPQP$,
$\EXPH=\BPEXP$, $\NEXP=\REXP$, and $\UEXP=\EXP$; the binary
comparator result separately retains ordinary linear advice.

The unconditional diagonal languages are
$H\in\Sigma_5 E$ and
$G\in\QH\setminus(\NP/\poly\cup\coNP/\poly)$.
The advice consequence of $\hyp$ gives $\PH\subseteq\NP/\poly$,
and hence $\PH\subsetneq\QH$ independently of the randomized
simulation. Applying that simulation to $G$ and using Adleman's
theorem instead gives the quantitative, advice-free separation
\[
 \mathrm{BPP}\subsetneq
 \BPTIME\!\left(2^{O(\Llog(n)^{d^2+\varepsilon})}\right)
 \qquad(\varepsilon>0\text{ fixed}).
\]
This clock remains subexponential under every fixed number of
self-compositions, the smaller-gap regime distinguished from the
Karpinski--Verbeek subexponential separation by
Lu--Oliveira--Santhanam~\cite{karpinski-verbeek87,lu-oliveira-santhanam21}.
The qualitative implication from $\QH\subseteq\BPQP$ also applies
under other hypotheses, including $\NP\subseteq\mathrm{BPP}$.
\Cref{rem:kannan-clock} retains exponent $d^2$ for a separate language
at each fixed polynomial circuit-size bound.

The new relativized lower bound complements the randomized SAT upper
bound. For each admissible $K$, one layered oracle distribution almost
surely gives a $K$-membership-comparable complete set,
$\NP^A=\coNP^A$, and a unary NP language outside
$\BPTIME^A(T)$ whenever $T(n)2^{-K(n)}\to0$.
For $K(n)=\lceil c_0\Llog(n)^d\rceil$, $d\ge2$, this rules out all
clocks $2^{\delta\Llog(n)^d}$ with $0<\delta<c_0$ and fixes the
logarithmic power $d$ among relativizing randomized bounds.
The one-length coupling controls adaptive queries to every table level;
separated input lengths turn that estimate into a single-oracle result.

\paragraph{Two randomized-comparator regimes.}
An inverse-polynomial relative saving $\gamma(n)\ge1/\poly(n)$ below
$2^{-t}$ gives the structural collapse and the level-independent
randomized PH exponent $d^2$, together with the exponent-$d$ symmetric
and polynomial-advice simulations. Limited independence supplies
succinct sample tables, and polynomially many occurrence indices certify
a heavy output. The resulting $\NP/\poly$ comparison relation feeds
the self-contained positive-relation theorem in
Appendix~\ref{app:positive-abg}. \Cref{thm:relative-gap-transfer}
therefore carries all conclusions of \cref{thm:main} except its uniform
deterministic promise-unique algorithm and $\UEXP=\EXP$.
For binary comparators, deterministic decoding uses polynomial advice;
one-sided randomized decoding keeps linear advice.

A saving only $2^{-O(\Llog(n)^d)}$ still gives one-sided randomized
SAT time $2^{O(\Llog(n)^d)}$, fixed-level PH exponent $d^k$,
$\QH=\BPQP$, $\EXPH=\BPEXP$, and $\NEXP=\REXP$. The fixed signed
existential diagonal language receives randomized exponent $d^5$;
a fourth-level deterministic diagonal gives the BPP-only separation
with exponent $d^4$. For $d>1$, the implication
$\Hwr\Rightarrow\NP\subseteq\coNP/\poly$ remains unresolved here.
The obstruction is explicit:
\[
 R=\Theta\bigl((b_n+1)/\gamma(n)^2\bigr),\qquad
 |\pi|=\Theta(R\log J).
\]
For quasipolynomially small $\gamma$, this occurrence certificate can
be quasipolynomially long even though the table description remains
polynomial. Polynomial-advice certified SAT search therefore survives,
but the positive-relation structural transfer does not follow. At
$d=1$ the weak saving is inverse-polynomial and this distinction
vanishes. The same deterministic oracle comparator covers both
regimes and limits their uniform randomized SAT exponent.

For $d>1$, the ordinary second-level target \eqref{eq:target} remains
unproved. The oracle has $\PH^A=\NP^A$, so it addresses the randomized
time improvement rather than that target. In the certified-search
route, a polynomial-time replacement with polynomial advice is
equivalent to $\SAT\in\Pclass/\poly$. The explicit longer predicate
clock is the remaining step in that route.

\paragraph{Preparation and provenance.}
This note was developed through an AI-assisted mathematical discussion.
The human participant supplied the certified-search, symmetric-alternation,
padding, and layered-oracle proposals, together with the attribution and
hierarchy-context revisions. Integration checked witness extraction,
common-length encodings, the oracle's level dependence and hybrid
independence, finite exceptional inputs, and the stated resource bounds.
The randomized-comparator extension was also supplied by the human
participant; integration checked mode amplification and the
randomness, certificate-length, and advice bounds. The relative-gap
supplement is integrated in \cref{sec:relative-gap}. The human
participant subsequently proposed replacing its sample-table generator
by limited independence and requested the complete positive-relation
proof. This revision implements that replacement, gives the finite-field
local evaluator and its even-moment tail bound, and proves the
cover-versus-residual transfer in Appendix~\ref{app:positive-abg} with
an explicit majority threshold and witness format. The randomized
results are organized into two regimes according to the occurrence-
certificate length. The imported results are cited where used. The arguments have not been formally verified in a
proof assistant.

\clearpage
\appendix
\section{A self-contained positive-relation advice transfer}
\label{app:positive-abg}

We prove \cref{lem:positive-abg} by the cover-versus-residual argument
of Amir, Beigel, and Gasarch~\cite[Theorem~4.4]{ABG03}, retaining the
certificates for every positive relation test. The proof gives a single
polynomial-length advice string for both membership and nonmembership;
it never asks the verifier to certify a negative answer about $W$.

\begin{proof}[Proof of \cref{lem:positive-abg}]
Fix $n\ge1$, put $X=\bits^n$, $k=k(n)$, and let
$a(x)=\chi_A(x)$. For tuples, $a(\mathbf x)$ denotes the vector of
coordinate labels. All choices below are nonuniform choices for this
one length. The case $n=0$ is handled by advising its one membership bit.

\paragraph{Normalize the NP/poly verifier for $W$.}
A fixed polynomial $Q(n)$ bounds the encoding lengths of every pair
$(\boldsymbol\tau,\boldsymbol\beta)$ with
$\boldsymbol\tau\in X^k$ and $\boldsymbol\beta\in\bits^k$.
Bundle the designated $W$-advice for every encoding length at most
$Q(n)$ into $\eta_n$. Pad certificates to a common polynomial bound
$w(n)$. There is then a deterministic polynomial-time predicate $V_W$
such that, for every pair under consideration,
\begin{equation}
 W(\boldsymbol\tau,\boldsymbol\beta)
 \iff \exists\omega\in\bits^{w(n)}\;
 V_W(n,\boldsymbol\tau,\boldsymbol\beta,\eta_n,\omega)=1.
 \label{eq:abg-np-normalization}
\end{equation}
The advice and clock are polynomial in $n$, since $k(n)$ has a fixed
polynomial upper bound. The value of $k(n)$ can itself be advised;
computability of the arity function is unnecessary for this lemma.
Every invocation below is a positive test of the left side of
\eqref{eq:abg-np-normalization} and will be implemented by guessing
$\omega$ and checking its right side.

\paragraph{Relations with a truthfully labeled suffix.}
Suppose finite anchor pools $Z_{m+1},\ldots,Z_k\subseteq X$ have been
chosen, together with their true labels. Define a relation on $m$-tuples
by
\begin{equation}
 \begin{split}
 W_m(\mathbf t,\boldsymbol\beta)\iff{}&
 \exists z_{m+1}\in Z_{m+1},\ldots,z_k\in Z_k:\\
 &W\bigl((\mathbf t,z_{m+1},\ldots,z_k),
        (\boldsymbol\beta,a(z_{m+1}),\ldots,a(z_k))\bigr).
 \end{split}
 \label{eq:abg-suffix-relation}
\end{equation}
At $m=k$ the suffix is empty and $W_k=W$. In all uses of
\eqref{eq:abg-suffix-relation}, a witness chooses one index in each
pool and one $V_W$-certificate. It does not enumerate the product of the
pools. Thus each $W_m$ has a polynomial-size NP certificate given the
pools and their labels.

We maintain the invariant that for every $\mathbf t\in X^m$,
\begin{equation}
 \exists\boldsymbol\beta\in\bits^m\ W_m(\mathbf t,\boldsymbol\beta),
 \qquad
 W_m(\mathbf t,\boldsymbol\beta)
       \Longrightarrow \boldsymbol\beta\ne a(\mathbf t).
 \label{eq:abg-suffix-invariant}
\end{equation}
It holds initially by the hypothesis on $W$. The exclusion implication
always follows from the true suffix labels; the construction below
preserves totality when decreasing $m$.

\paragraph{The finite covering construction.}
At a stage $m\ge2$ satisfying \eqref{eq:abg-suffix-invariant}, for each
$z\in X$ define
\[
 D_m(z)=\{\mathbf t\in X^{m-1}:
      \exists\boldsymbol\beta\in\bits^{m-1}\;
      W_m((\mathbf t,z),(\boldsymbol\beta,a(z)))\}.
\]
Start with $T=X^{m-1}$ and $Z_m=\varnothing$. While $T\ne\varnothing$
and some $z$ satisfies $|D_m(z)\cap T|\ge|T|/4$, add $z$ to $Z_m$ and
replace $T$ by $T\setminus D_m(z)$. Each iteration removes at least one
quarter of the residual set. Thus
\begin{equation}
 |Z_m|\le4n(m-1)+1.
 \label{eq:abg-pool-size}
\end{equation}
Indeed, after $j$ iterations the residual size is at most
$2^{n(m-1)}(3/4)^j$, which is less than one by the displayed bound.
A previously selected $z$ cannot cover a nonempty later residual set,
so repetitions of anchors are unnecessary.

If $T=\varnothing$, every $(m-1)$-tuple has a covering anchor in $Z_m$.
Using that anchor in \eqref{eq:abg-suffix-relation} proves totality of
$W_{m-1}$, and its exclusion property follows from the true suffix.
Continue at stage $m-1$. Otherwise stop, retaining this nonempty $T$;
then
\begin{equation}
 |D_m(x)\cap T|<|T|/4\qquad\text{for every }x\in X.
 \label{eq:abg-residual-property}
\end{equation}
The pools, construction choices, and membership labels are used only to
establish existence of advice. The final verifier does not run this
construction or test the residual property.

\paragraph{Case 1: all stages are covered.}
We reach the total one-coordinate relation $W_1$. Since it excludes
$a(x)$, it accepts precisely the other bit. On input $x$ and requested
membership bit $b\in\{0,1\}$, the verifier guesses
$z_j\in Z_j$ for $2\le j\le k$ and a certificate for
\begin{equation}
 W\bigl((x,z_2,\ldots,z_k),
        (1-b,a(z_2),\ldots,a(z_k))\bigr).
 \label{eq:abg-exhausted-verifier}
\end{equation}
The relation's totality supplies such a certificate if $b=a(x)$;
exclusion of the true vector prevents one if $b\ne a(x)$.
This also handles $k=1$ directly, with no pools or covering stages.

\paragraph{Case 2: a nonempty residual set remains.}
Suppose the construction stops at stage $m\ge2$. For each $x\in X$
put $T_x=T\setminus D_m(x)$. Equation~\eqref{eq:abg-residual-property}
gives $|T_x|>3|T|/4$. On a tuple $\mathbf u\in T_x$, the invariant
makes $W_m((\mathbf u,x),\cdot)$ total, but the definition of $D_m(x)$
forbids every accepted vector whose last bit is $a(x)$. Hence, for every $\mathbf u\in T_x$ and both values of $b$,
\begin{equation}
 \left[
 \exists\boldsymbol\beta\in\bits^{m-1}\;
 W_m((\mathbf u,x),(\boldsymbol\beta,1-b))
 \right]\iff b=a(x).
 \label{eq:abg-good-position}
\end{equation}
This conclusion holds simultaneously for all choices of the NP
witnesses: on a good position, the wrong requested bit has no witness.

Take
\begin{equation}
 M=16(n+1)+1,\qquad h=(M+1)/2.
 \label{eq:abg-majority-threshold}
\end{equation}
Choose $M$ independent uniform elements $\mathbf u_1,\ldots,\mathbf u_M$
of the finite set $T$, allowing repetitions. For any fixed $x$, each
position belongs to $T_x$ with probability greater than $3/4$.
The elementary additive Chernoff bound gives
\[
 \Pr\bigl[\#\{i:\mathbf u_i\in T_x\}<h\bigr]
 \le e^{-M/8}.
\]
Since $2^n e^{-M/8}<1$, a union bound fixes one list for which
\begin{equation}
 \#\{i\in[M]:\mathbf u_i\in T_x\}\ge h
       \qquad\text{for every }x\in X.
 \label{eq:abg-fixed-majority}
\end{equation}
Only this list, not the residual set $T$, is stored as advice. Sampling
from $T$ is an existence argument and need not be efficient.

\paragraph{The majority verifier and its exact witness.}
For input $(x,b)$ the verifier guesses $h$ distinct \emph{positions}
\[
 1\le i_1<i_2<\cdots<i_h\le M.
\]
For each selected position $i_j$, it guesses a vector
$\boldsymbol\beta^{(j)}\in\bits^{m-1}$, one anchor
$z_{m+1}^{(j)}\in Z_{m+1},\ldots,z_k^{(j)}\in Z_k$, and one
$V_W$-certificate for
\begin{equation}
 \begin{split}
 W\bigl(
   (\mathbf u_{i_j},x,z_{m+1}^{(j)},\ldots,z_k^{(j)}),
   (\boldsymbol\beta^{(j)},1-b,
       a(z_{m+1}^{(j)}),\ldots,a(z_k^{(j)}))
 \bigr).
 \end{split}
 \label{eq:abg-majority-verifier}
\end{equation}
It accepts exactly when all $h$ certificates pass the deterministic
check in \eqref{eq:abg-np-normalization}. Positions must be distinct;
the advised tuples at different positions need not be.

If $b=a(x)$, at least $h$ positions are good by
\eqref{eq:abg-fixed-majority}, and each supplies an accepting certificate
by \eqref{eq:abg-good-position}. If $b\ne a(x)$, no good position can
have such a certificate. There are at most $M-h=h-1$ bad positions,
so no selection of $h$ distinct positions can pass. This proves
completeness and soundness of the majority test, rather than assuming
an oracle for the truth values of its NP predicates.

\paragraph{Advice, witness length, and the two sides.}
The advice stores $k$, a case flag, the stopping index if applicable,
$\eta_n$, the anchor pools and their true labels, and, in Case~2, the
fixed list of $M$ tuples. By \eqref{eq:abg-pool-size}, the pools and
labels use $O((n+1)^2k^2)$ bits. The list uses
$Mn(m-1)=O((n+1)^2k)$ bits. Both are polynomial, as is $\eta_n$.
Each guessed anchor is a position in an advised pool. The majority
witness has $h=O(n+1)$ copies of a polynomial-length relation witness,
a $(m-1)$-bit vector, and at most $k-m$ anchor indices; it too has
polynomial length. Every check, including distinctness of positions
and evaluation of the $h$ relation certificates, is polynomial-time.
All formats and clocks have fixed polynomial bounds on arbitrary
candidate advice; correctness uses the designated advice just constructed.

We have one common advice sequence and a verifier accepting $(x,b)$
exactly when $b=a(x)$ has a witness. Fixing $b=1$ proves
$A\in\NP/\poly$, and fixing $b=0$ proves
$\overline A\in\NP/\poly$. Thus
$A\in\NP/\poly\cap\coNP/\poly$, as required.
\end{proof}


\begin{thebibliography}{99}
\small
\setlength{\itemsep}{2pt}

\bibitem{Siv99}
D.~Sivakumar.
\newblock On membership comparable sets.
\newblock \emph{Journal of Computer and System Sciences}, 59(2):270--280,
1999.
\newblock \href{https://doi.org/10.1006/jcss.1999.1653}{doi:10.1006/jcss.1999.1653}.

\bibitem{BD13}
Sebastian Ben Daniel.
\newblock \emph{The Complexity of Direct-Sum Questions}.
\newblock Ph.D. thesis, Ben-Gurion University of the Negev, March~1, 2013.
\newblock Chapter~5, especially Theorem~5.2.1 and Section~5.3.
\newblock National Library of Israel catalogue record:
\url{https://www.nli.org.il/en/dissertations/NNL_ALEPH990036909080205171/NLI}.

\bibitem{BD26}
Sebastian Ben Daniel.
\newblock \emph{Linear Certificates for Membership Comparability,
Quadratic Barriers for Selectors}.
\newblock \href{https://arxiv.org/abs/2609.31053}{arXiv:2609.31053} [cs.CC], 2026.
\newblock Theorem~3.1, Theorem~7.1, and Appendices~F and~H.

\bibitem{ABG03}
Amihood Amir, Richard Beigel, and William Gasarch.
\newblock Some connections between bounded query classes and non-uniform
complexity.
\newblock \emph{Information and Computation}, 186(1):104--139, 2003.
\newblock \href{https://doi.org/10.1016/S0890-5401(03)00091-9}{doi:10.1016/S0890-5401(03)00091-9}.
\newblock Theorem~4.4 is cited using the authors' long version,
\url{https://www.cs.umd.edu/~gasarch/papers/nonuniform.pdf}.

\bibitem{CCHO05}
Jin-Yi Cai, Venkatesan~T. Chakaravarthy, Lane~A. Hemaspaandra,
and Mitsunori Ogihara.
\newblock Competing provers yield improved Karp--Lipton collapse results.
\newblock \emph{Information and Computation}, 198(1):1--23, 2005.
\newblock \href{https://doi.org/10.1016/j.ic.2005.01.002}{doi:10.1016/j.ic.2005.01.002}.
\newblock Earlier technical-report version: University of Rochester,
UR CSD/TR759, revised November 2002,
\url{https://hdl.handle.net/1802/594}.

\bibitem{VV86}
Leslie~G. Valiant and Vijay~V. Vazirani.
\newblock NP is as easy as detecting unique solutions.
\newblock \emph{Theoretical Computer Science}, 47:85--93, 1986.
\newblock \href{https://doi.org/10.1016/0304-3975(86)90135-0}{doi:10.1016/0304-3975(86)90135-0}.

\bibitem{Sud96}
Madhu Sudan.
\newblock Maximum likelihood decoding of Reed--Solomon codes.
\newblock In \emph{Proceedings of the 37th Annual IEEE Symposium on
Foundations of Computer Science}, pages 164--172, 1996.
\newblock The finite-field reconstruction statement used here is the
form recorded as Lemma~4 in~\cite{Siv99}.

\bibitem{Sauer72}
Norbert Sauer.
\newblock On the density of families of sets.
\newblock \emph{Journal of Combinatorial Theory, Series A}, 13:145--147,
1972.

\bibitem{cai-s2zpp}
Jin-Yi Cai.
\newblock $S_2^p\subseteq\ZPP^{\NP}$.
\newblock \emph{Journal of Computer and System Sciences}, 73(1):25--35,
2007.
\newblock \href{https://doi.org/10.1016/j.jcss.2003.07.015}{doi:10.1016/j.jcss.2003.07.015}.

\bibitem{kannan82}
Ravi Kannan.
\newblock Circuit-size lower bounds and non-reducibility to sparse sets.
\newblock \emph{Information and Control}, 55(1--3):40--56, 1982.
\newblock \href{https://doi.org/10.1016/S0019-9958(82)90382-5}{doi:10.1016/S0019-9958(82)90382-5}.

\bibitem{chen-hirahara-ren24}
Lijie Chen, Shuichi Hirahara, and Hanlin Ren.
\newblock Symmetric exponential time requires near-maximum circuit size.
\newblock In \emph{Proceedings of the 56th Annual ACM Symposium on Theory
of Computing (STOC 2024)}, pages 1990--1999. ACM, 2024.
\newblock \href{https://doi.org/10.1145/3618260.3649624}{doi:10.1145/3618260.3649624}.
\newblock Author-hosted proceedings version:
\url{https://hanlin-ren.github.io/files/pdf/stoc24_S2E_lb.pdf}.

\bibitem{li24}
Zeyong Li.
\newblock Symmetric exponential time requires near-maximum circuit size:
simplified, truly uniform.
\newblock In \emph{Proceedings of the 56th Annual ACM Symposium on Theory
of Computing (STOC 2024)}, pages 2000--2007. ACM, 2024.
\newblock \href{https://doi.org/10.1145/3618260.3649615}{doi:10.1145/3618260.3649615}.
\newblock Revised author manuscript: ECCC TR23-156, revision~1,
April~4, 2024, \url{https://eccc.weizmann.ac.il/report/2023/156/}.

\bibitem{bft98}
Harry Buhrman, Lance Fortnow, and Thomas Thierauf.
\newblock Nonrelativizing separations.
\newblock In \emph{Proceedings of the 13th Annual IEEE Conference on
Computational Complexity}, pages 8--12. IEEE Computer Society, 1998.
\newblock \href{https://doi.org/10.1109/CCC.1998.694585}{doi:10.1109/CCC.1998.694585}.
\newblock Author-hosted manuscript:
\url{https://lance.fortnow.com/papers/files/nonrel.pdf}.


\bibitem{Adl78}
Leonard~M. Adleman.
\newblock Two theorems on random polynomial time.
\newblock In \emph{Proceedings of the 19th Annual Symposium on Foundations
of Computer Science}, pages 75--83. IEEE Computer Society, 1978.
\newblock \href{https://doi.org/10.1109/SFCS.1978.37}{doi:10.1109/SFCS.1978.37}.

\bibitem{Bar02}
Boaz Barak.
\newblock A probabilistic-time hierarchy theorem for ``slightly
non-uniform'' algorithms.
\newblock In Jos\'e D.~P. Rolim and Salil Vadhan, editors,
\emph{Randomization and Approximation Techniques in Computer Science
(RANDOM 2002)}, volume 2483 of \emph{Lecture Notes in Computer Science},
pages 194--208. Springer, 2002.
\newblock \href{https://doi.org/10.1007/3-540-45726-7_16}{doi:10.1007/3-540-45726-7\_16}.

\bibitem{FS04}
Lance Fortnow and Rahul Santhanam.
\newblock Hierarchy theorems for probabilistic polynomial time.
\newblock In \emph{Proceedings of the 45th Annual IEEE Symposium on
Foundations of Computer Science}, pages 316--324. IEEE Computer Society,
2004.
\newblock \href{https://doi.org/10.1109/FOCS.2004.33}{doi:10.1109/FOCS.2004.33}.
\newblock Author-hosted manuscript:
\url{https://lance.fortnow.com/papers/files/probhier.pdf}.

\bibitem{He25}
Songhua He.
\newblock A note on a hierarchy theorem for promise-BPTIME.
\newblock \emph{Electronic Colloquium on Computational Complexity},
TR25-004, revision~1, December~5, 2025.
\newblock \url{https://eccc.weizmann.ac.il/report/2025/004/}.


\bibitem{karpinski-verbeek87}
Marek Karpinski and Rutger Verbeek.
\newblock Randomness, provability, and the separation of Monte Carlo time
and space.
\newblock In Egon B\"orger, editor, \emph{Computation Theory and Logic},
volume 270 of \emph{Lecture Notes in Computer Science}, pages 189--207.
Springer, 1987.
\newblock \href{https://doi.org/10.1007/3-540-18170-9_166}{doi:10.1007/3-540-18170-9\_166}.

\bibitem{lu-oliveira-santhanam21}
Zhenjian Lu, Igor~C. Oliveira, and Rahul Santhanam.
\newblock Pseudodeterministic algorithms and the structure of probabilistic time.
\newblock In \emph{Proceedings of the 53rd Annual ACM SIGACT Symposium on
Theory of Computing (STOC 2021)}, pages 303--316. ACM, 2021.
\newblock \href{https://doi.org/10.1145/3406325.3451085}{doi:10.1145/3406325.3451085}.
\newblock Author manuscript: \url{https://arxiv.org/abs/2103.08539}.

\bibitem{bennett-gill81}
Charles~H. Bennett and John Gill.
\newblock Relative to a random oracle $A$,
$\Pclass^A\ne\NP^A\ne\coNP^A$ with probability~1.
\newblock \emph{SIAM Journal on Computing}, 10(1):96--113, 1981.
\newblock \href{https://doi.org/10.1137/0210008}{doi:10.1137/0210008}.

\bibitem{SSS95}
Jeanette~P. Schmidt, Alan Siegel, and Aravind Srinivasan.
\newblock Chernoff--Hoeffding bounds for applications with limited independence.
\newblock \emph{SIAM Journal on Discrete Mathematics}, 8(2):223--250, 1995.
\newblock \href{https://doi.org/10.1137/S089548019223872X}{doi:10.1137/S089548019223872X}.
\newblock Author-hosted manuscript:
\url{https://www.cs.umd.edu/~srin/PDF/ch-bounds.pdf}.

\bibitem{BEG17}
Salman Beigi, Omid Etesami, and Amin Gohari.
\newblock The value of help bits in randomized and average-case complexity.
\newblock \emph{Computational Complexity}, 26:119--145, 2017.
\newblock \href{https://doi.org/10.1007/s00037-016-0135-x}{doi:10.1007/s00037-016-0135-x}.
\newblock Preprint: \url{https://arxiv.org/abs/1408.0499}.

\end{thebibliography}
\end{document}